\documentclass[runningheads]{llncs}

\usepackage[T1]{fontenc}

\usepackage{amsmath}
\usepackage{amssymb}
\usepackage{enumerate}
\usepackage{color}
\usepackage{cite}
\usepackage{booktabs}
\usepackage{marvosym}
\usepackage{definitions}
\usepackage[hypertexnames=false]{hyperref}

\begin{document}

\title{Unification of Deterministic Higher-Order Patterns}

\titlerunning{Unification of Deterministic Higher-Order Patterns}

\author{Johannes Niederhauser\textsuperscript{(\Letter)}%
\orcidID{0000-0002-8662-6834} \and
Aart Middeldorp\orcidID{0000-0001-7366-8464}}

\authorrunning{J.~Niederhauser and A.~Middeldorp}

\institute{Department of Computer Science,
University of Innsbruck, Innsbruck, Austria
\email{\{johannes.niederhauser,aart.middeldorp\}@uibk.ac.at}}

\maketitle

\begin{abstract}
We present a sound and complete unification procedure for
deterministic higher-order patterns, a class of simply-typed
lambda terms introduced by Yokoyama et al.~which comes
with a deterministic matching problem. Our unification procedure
can be seen as a special case of full higher-order unification
where flex-flex pairs can be solved in a most general way.
Moreover, our method generalizes Libal and Miller's recent
functions-as-constructors higher-order unification (FCU)
by dropping their global restriction on variable arguments,
thereby losing the property that every solvable problem has a
most general unifier. In fact, minimal complete sets of
unifiers of deterministic higher-order patterns may be
infinite, so decidability of the unification problem
remains an open question.
\end{abstract}

\section{Introduction}
\label{sec:introduction}

Given two simply-typed $\lambda$-terms $s$ and $t$, higher-order
unification tries to answer the question whether they can be made
equivalent up to $\beta\eta$-conversion by applying some substitution
$\theta$ to both of them, i.e., $s\theta \LRab{\beta\eta}{*} t\theta$.
The first practical higher-order unification procedure has been
introduced by Huet in 1975 \cite{H75}. Since then, higher-order
unification has evolved into an extensively studied topic which plays
a central role in research areas such as higher-order automated
theorem proving \cite{A+96,SB21,BBTV23}, interactive theorem proving
\cite{NPW02}, higher-order logic programming \cite{M91} and
higher-order rewriting \cite{MN98}. In contrast to first-order
unification, which is decidable and unitary (most general unifiers
exist), higher-order unification is undecidable and infinitary
(infinitely many incomparable unifiers may exist) \cite{D01}. Since
neither of these properties makes higher-order unification a
well-behaved subroutine for tasks such as theorem proving, more
well-behaved subclasses of higher-order unification have been
identified and investigated in the past. One important subclass is
higher-order pattern unification due to Miller \cite{M91} where the
arguments of
free variables are restricted to $\eta$-equivalent representations of
distinct bound variables. Even though this rules out many simply-typed
$\lambda$-terms, the unification procedure is sufficient for many
practical purposes \cite{MN98,M91} while being both decidable and
unitary. Furthermore, higher-order pattern unification has linear
time complexity \cite{Q93}.

\begin{example}[taken from \cite{N93}]
Let $\m{c}$ be a constant and $F$, $G$ be variables. The
higher-order patterns $\lambda x,y.F\:x$ and
$\lambda x,y.\m{c}\:(G\:y\:x)$ are unifiable with
$\theta = \SEt{F \mapsto \lambda z.\m{c}\:(H\:z),
G \mapsto \lambda z_1,z_2.H\:z_2}$ where $H$ is a fresh variable:
\[
(\lambda x,y.F\:x)\theta =_\beta \lambda x,y.\m{c}\:(H\:x) =_\beta
(\lambda x,y.\m{c}\:(G\:y\:x))\theta
\]
In fact, $\theta$ is a most general unifier.
\end{example}

Recently, Libal and Miller have generalized higher-order pattern
unification to
functions-as-constructors higher-order unification (FCU) \cite{LM22}.
In this subclass, arguments of free variables may contain constants.
In contrast to Miller's patterns, it is not enough to demand
distinctness of arguments of a free variable. Instead, no argument
of a free variable may be a higher-order subterm of another argument
of the same free variable.
In addition, it is required that an argument of a free variable
is not a strict higher-order subterm of any argument of any
variable in the unification problem (global restriction)
in order to retain unicity and decidability of the unification problem.

\begin{example}[taken from \cite{LM22}]
Let $\m{fst}$, $\m{snd}$ be constants and $X$, $Y$ be variables. The
simply-typed $\lambda$-terms
$\lambda l_1,l_2.X\:(\m{fst}\:l_1)\:(\m{fst}\:(\m{snd}\:l_1))$ and
$\lambda l_1,l_2.\m{snd}\:(Y\:(\m{fst}\:l_2)\:(\m{fst}\:l_1))$ are
unifiable with
$\theta = \SEt{X \mapsto \lambda z_1,z_2.\m{snd}\:(H\:z_1),
Y \mapsto \lambda z_1,z_2.H\:z_2}$ where $H$ is a fresh variable:
\begin{align*}
(\lambda l_1,l_2.X\:(\m{fst}\:l_1)\:(\m{fst}\:(\m{snd}\:l_1)))\theta
&=_\beta
\m{snd}\:(H\:(\m{fst}\:l_1)) \\ &=_\beta
(\lambda l_1,l_2.\m{snd}\:(Y\:(\m{fst}\:l_2)\:(\m{fst}\:l_1)))\theta
\end{align*}
In fact, $\theta$ is a most general unifier.
\end{example}

Note that unlike higher-order pattern unification, FCU does not give
rise to a corresponding class of simply-typed lambda terms: Due to the
global restriction, the question whether
$\lambda l_1,l_2.X\:(\m{fst}\:l_1)\:(\m{fst}\:(\m{snd}\:l_1))$
is in the scope of FCU depends on the remaining terms in the
unification problem. In our previous example, it worked, but if we
had chosen
$\lambda l_1,l_2.\m{snd}\:(Y\:(\m{fst}\:l_2)\:(\m{snd}\:l_1))$
(which is also fine on its own) as the second term, FCU could not be
used as $\m{snd}\:l_1$ is a higher-order subterm of
$(\m{fst}\:(\m{snd}\:l_1))$. In \cite{LM22}, FCU is compared to
Yokoyama et al.'s deterministic second-order patterns (DSPs), a class of
simply-typed
$\lambda$-terms which have a deterministic matching problem, i.e.,
given a DSP $s$ and a simply-typed $\lambda$-term $t$, there is at
most one substitution $\theta$ such that $s\theta = t$ \cite{YHT04b}.
Even though it is not discussed in \cite{LM22},
Yokoyama et al.~also introduced deterministic higher-order
patterns (DHPs) \cite{YHT04} which is a generalization of
DSPs to terms of any finite order. DHPs are very
similar to the terms used in FCU, but since one is only interested in
matching, the global restriction of FCU is not needed. Furthermore,
the set of DHPs strictly includes the set of terms which can in
principle be used in FCU. Since there is no global restriction
as in FCU, unification of DHPs is not unitary.

\begin{example}[taken from \cite{YHT04}]
\label{exa:dhp}
Let $\m{f}$ be a constant and $M$, $N$ be variables. The DHPs
$\lambda x,y.M\:(\m{f}\:x)\:(\m{f}\:y)$ and
$\lambda x,y.\m{f}\:(N\:y\:x)$ admit the three unifiers
\begin{enumerate}[(i)]
\item
$\SEt{M \mapsto \lambda z_1,z_2.z_1, N \mapsto \lambda z_1,z_2.z_2}$,
\item
$\SEt{M \mapsto \lambda z_1,z_2.z_2, N \mapsto \lambda z_1,z_2.z_1}$, and
\item
$\SET{M \mapsto \lambda z_1,z_2.\m{f}\:(Z\:z_1\:z_2),
N \mapsto \lambda z_1,z_2.Z\:(\m{f}\:z_2)\:(\m{f}\:z_1)}$,
\end{enumerate}
which constitute a complete set of unifiers. Here, $Z$ is a fresh
variable.
\end{example}

We think that it is worthwhile to investigate unification of DHPs
because it occupies a special position where matching is still
deterministic but most general unifiers do not exist
anymore. Interestingly, instead of moving from unitary to finitary
unification problems, we immediately jump to infinitary problems as we
will show in \exaref{infinite}. Another reason for the usage
of DHP unification is that the definition of DHPs
as a class of simply-typed $\lambda$-terms makes it easier to use it as a
subroutine than FCU since one cannot get in a situation where the
problem is undefined due to the global restriction.

The remainder of the paper is structured as follows: After
setting the scene in \secref{preliminaries}, we define
DHPs in \secref{dhp} and introduce an inference system to solve
their unification problem in \secref{unification}.
Correctness of the inference system is proved in \secref{correct}.
A sufficient criterion to get minimal solutions is
described in \secref{minimal} before we conclude in \secref{conclusion}.

\section{Preliminaries}
\label{sec:preliminaries}

Throughout this text, we will denote a sequence
$\seq{a}$ by $\vec{a}$ where $n \geq 0$ and the corresponding set by
$\SEt{\vec{a}}$. For every $a$, $\vec[0]{a}$ represents the empty
sequence which we denote by $()$. We implicitly enclose sequences in
parentheses whenever it is necessary. For a binary relation $R$,
$\vec{a} \mathrel{R} \vec{b}$ abbreviates
$a_1 \mathrel{R} b_1, \dots, a_n \mathrel{R} b_n$. Similarly,
$f(\vec{a})$ denotes the pointwise application of the function $f$ to
$\vec{a}$. Given two sequences $\vec{a}$ and $\vec[m]{b}$, their
concatenation is written as $\vec{a},\vec[m]{b}$. We write
$R^*$ for the transitive-reflexive closure of a binary relation $R$.

In this paper, we consider simply-typed $\lambda$-terms
\cite{C40,B92}. Let $\xS$ be a set of \emph{sorts} ($a$, $b$). We use a
flattened representation of simple types, so the set $\xT$ of
\emph{types} ($\sigma$, $\tau$, $\varphi$) is defined as follows:
$\xS \subseteq \xT$ and if $\seq{\sigma} \in \xT$ and $a \in \xS$ then
$\vec{\sigma} \to a \in \xT$ where we identify $() \to a$ with
$a$. The \emph{order} of a type is defined as follows:
$\ord{\iota} = 1$ and
$\ord{\vec{\sigma} \to a} = \m{max}(\ord{\vec{\sigma}}) + 1$.

Since our unification procedure is defined for terms in
$\beta\eta$-long normal form, we will use a $\beta\eta$-free
formulation of terms and substitutions based on \cite{F23}. This way,
we can specify substitution application without referring to
$\beta$-reduction which simplifies the technical development
in the proofs of this paper when compared to similar works
such as \cite{SG89}. Suppose there is an infinite set $\xV$ of typed
variables ($x$, $y$, $z$, $w$, $F$, $G$, $H$, $M$, $N$) and a set
$\xF$ of typed function symbols ($c$, $d$, $f$, $g$, $h$)
such that $\xV$, $\xF$, $\xS$ and $\SET{\to}$ are
disjoint and there are infinitely many variables of each type.
Uppercase variable names typically represent free variables in examples,
but lowercase variable names may also be used for free variables. A
\emph{head} $h$ is either a function symbol or a variable and we write
$h : \sigma$ to denote that it has type $\sigma$.
In order to keep the notation concise, we drop the symbol $\lambda$
in abstractions. The following
inference rule defines the set $\trm{\sigma}$ of \emph{terms} ($s$,
$t$, $u$, $v$, $p$, $q$, $r$) of type $\sigma$
\begin{gather*}
\frac{h : \vec[m]{\tau} \to a \in \xF \cup \xV \quad
\vec[m]{t} \in \trm{\vec[m]{\tau}} \quad
\vec{x} : \vec{\sigma} \in \xV^n}
{\vec{x}.h(\vec[m]{t}) \in \trm{a}}
\end{gather*}
where $\trm{\vec[m]{\tau}}$ is a shorthand for
$\trm{\tau_1} \times \cdots \times \trm{\tau_m}$.
We abbreviate
$h()$ by $h$, $().s$ by $s$ and follow the convention that arbitrary
postfix operations $\diamondsuit$ bind stronger than binders in terms, so
$\vec{x}.t\diamondsuit = \vec{x}.(t\diamondsuit)$. Moreover,
$\vec[m]{\vec{x}.s\diamondsuit}$ denotes the list of terms
$\vec{x}.s_1\diamondsuit,\dots,\vec{x}.s_m\diamondsuit$. Furthermore,
we write $\hd{s}$ for the head of a term $s$. The set $\trm{\xS}$
denotes the set of all terms which coincides with the set of
well-typed lambda terms over $\xF$ and $\xV$ using sorts from $\xS$ in
$\beta\eta$-long normal form. As for function symbols and variables,
we denote that a term $s$ has type $\sigma$ by $s : \sigma$. Given a
term $s$, its set of free variables $\fv{s}$ is defined as
usual. Given some structure $A$ containing terms, we often write
$\fv{A}$ for the union of the sets of free variables of all terms in
the structure. A term $t$ is a \emph{subterm} of $s$ if $s \subteq t$
which is defined as follows: $s = \vec{x}.h(\vec[m]{s}) \subteq t$ if
$s = t$ or $\vec{x}.s_i \subteq t$ for some $1 \leq i \leq m$. We
define $s \subt t$ as $s \subteq t$ and $s \neq t$. The size $|s|$ of
a term is defined inductively as
$|\vec[k]{x}.h(\vec{t})| = 1 + \sum_{1 \leq i \leq n} t_i$.

We view terms modulo renaming of bound variables
($\alpha$-renaming). Hence, we may assume that no variable occurs both
free and bound in any term and that all bound variables are named
differently. Note that $\xV \subseteq \trm{\xS}$ does not hold as
terms are always in canonical form. The function $\cf$ takes a
variable and returns its canonical form by performing
$\eta$-expansion: $x\cf = \vec{y}.x(\vec{y}\cf)$ whenever
$x : \vec{\sigma} \to a$, $\vec{y} : \vec{\sigma}$ and
$x \notin \SEt{\vec{y}}$.

Finite mappings from variables to terms of the same type are called
\emph{substitutions} ($\theta, \gamma$, $\delta$, $\mu$, $\nu$). Given
a substitution $\theta = \SEt{\vec{x} \mapsto \vec{t}}$, its
\emph{domain} and \emph{image} are defined as
$\dom{\theta} = \SEt{\vec{x}}$ and $\im{\theta} = \SEt{\vec{t}}$.
The \emph{free variables} introduced by a substitution
$\theta$ are defined as the set
$\fv{\theta} = \bigcup\,\SET{\fv{t} \mid t \in \im{\theta}}$.
A set of variables $X$ is \emph{fresh} for a substitution $\theta$ if
$(\dom{\theta} \cup \im{\theta}) \cap X = \varnothing$. A
substitution $\theta$ is \emph{idempotent} if
$\dom{\theta} \cap \fv{\theta} = \varnothing$. Since terms are always
in their canonical form, the application of a substitution $\theta$ to
a term $t$ (written in postfix notation $t\theta$) is defined
hereditarily by implicitly performing $\beta$-reduction \cite{HL07}:
$x(\vec{t})\theta = u\SET{\vec{x} \mapsto \vec{t}\theta}$ if
$x \mapsto \vec{x}.u \in \theta$,
$h(\vec{t})\theta = h(\vec{t}\theta)$ if $h \notin \dom{\theta}$ and
$(\vec{x}.t)\theta = \vec{y}.t\SEt{\vec{x} \mapsto \vec{y}}\theta$
where $\SEt{\vec{y}}$ is fresh for $\theta$.
Hence, $t\theta$ represents the capture-avoiding application of
$\theta$ to the free variables of $t$. Note that substitution application
$s\theta$ is well-defined by induction on $(\ord{\theta},s)$ with respect
to $({>_\mathbb{N}},{\subt})$. We say that a term $s$
\emph{matches} a term $t$ if there exists a substitution $\theta$ such
that $s\theta = t$. Moreover, if $s\theta = t\theta$ then $\theta$ is
a \emph{unifier} of $s$ and $t$. Given
$\theta = \SEt{\vec{x} \mapsto \vec{s}}$ and
$\delta = \SEt{\vec[m]{y} \mapsto \vec[m]{t}}$, we define their
composition as
$\theta\delta = \SET{\vec{x} \mapsto \vec{s}\delta} \cup
\SET{y_j \mapsto t_j \mid \text{$1 \leq j \leq m$ and
$y_j \notin \SEt{\vec{x}}$}}$.

\begin{lemma}
\label{lem:substAppAssoc}
If $\theta$ and $\delta$ are substitutions then
$(s\theta)\delta = s(\theta\delta)$ for all terms $s$.
\end{lemma}

\begin{proof}
We proceed by induction on the definition of substitution application in
$s\theta$. Let
$s = \vec{x}.h(\vec[m]{s})$. Without loss of generality, we may assume
that $\SEt{\vec{x}}$ is fresh for $\theta$ and $\delta$.
Hence, $(s\theta)\delta = \vec{x}.(h(\vec[m]{s})\theta)\delta$. If
$h \notin \dom{\theta}$ then
$(s\theta)\delta = \vec{x}.h(\vec[m]{s}\theta)\delta$. If also
$h \notin \dom{\delta}$ then
$(s\theta)\delta = \vec{x}.h((\vec[m]{s}\theta)\delta) =
\vec{x}.h(\vec[m]{s}(\theta\delta)) = s(\theta\delta)$ where we use the
induction hypothesis $m$ times for the second equality. If
$(h \mapsto \vec[m]{y}.u) \in \delta$ then
$(s\theta)\delta = \vec{x}.u\SET{\vec[m]{y} \mapsto
(\vec[m]{s}\theta)\delta} = \vec{x}.u\SET{\vec[m]{y} \mapsto
\vec[m]{s}(\theta\delta)} = s(\theta\delta)$ where we use the
induction hypothesis $m$ times for the second equality. In the remaining
case, $(h \mapsto \vec[m]{y}.t) \in \theta$. Without loss of generality,
$\SEt{\vec[m]{y}}$ is fresh for $\delta$, so
\begin{align*}
(s\theta)\delta
&= \vec{x}.(t\SET{\vec[m]{y} \mapsto \vec[m]{s}\theta})\delta
= \vec{x}.t(\SET{\vec[m]{y} \mapsto \vec[m]{s}\theta}\delta) \\
&= \vec{x}.t\delta\SET{\vec[m]{y} \mapsto (\vec[m]{s}\theta)\delta}
= \vec{x}.t\delta\SET{\vec[m]{y} \mapsto \vec[m]{s}(\theta\delta)}
= s(\theta\delta)
\end{align*}
where the second and fourth equalities follow from the induction
hypothesis. \qed
\end{proof}

\lemref{substAppAssoc} allows us to write $s\theta\delta$ without
parentheses. Moreover, it implies associativity of substitution
composition. Given a set of variables $W$, we say that two
substitutions $\theta_1$ and $\theta_2$ are \emph{equal over W}
($\theta_1 = \theta_2 \restr{W}$) if $\theta_1(x) = \theta_2(x)$ for
all $x \in W$. A substitution $\theta$ is \emph{at least as general}
as a substitution $\delta$ with respect to $W$, denoted by
$\theta \subsumes \delta\restr{W}$, if there exists a substitution
$\gamma$ such that $\theta\gamma = \delta\restr{W}$. Associativity of
substitution composition implies transitivity of $\subsumes$. We drop
the suffix $[W]$ if $W$ is the set of all variables. Finally, we say
that two substitutions $\theta$ and $\delta$ are \emph{orthogonal}
($\theta \perp \delta$) if there is no substitution $\gamma$ such that
$\theta \subsumes \gamma$ and $\delta \subsumes \gamma$.

\section{Deterministic Higher-Order Patterns}
\label{sec:dhp}

The goal of DHPs is to approximate the class of terms which
have a deterministic matching problem while keeping the definition
simple and the corresponding matching procedure efficient. In the
following examples, we leave out type information as it is not
important.

\begin{example}
Consider the term $s = F(G)$. Given a term $t$ there cannot be at
most one substitution $\theta_t$ such that $s\theta_t = t$ as it is
not clear how to split the substitution between $F$ and $G$. For
example, if $t = \m{f}(x)$ then we have $s\theta_1 = t$ and
$s\theta_2 = t$ for
$\theta_1 = \SET{F \mapsto z.\m{f}(x), G \mapsto y}$ and
$\theta_2 = \SET{F \mapsto z.\m{f}(z), G \mapsto x}$. In fact, it is
easy to see that there are actually infinitely many possible
matching substitutions.
\end{example}

Hence, nested free variables will not be allowed in DHPs. On the other
hand, arguments of free variables have to contain bound variables as
the following example shows.

\begin{example}
Consider the term $s = F(\m{c})$. Again, matching is not
deterministic: For $t = \m{f}(\m{c})$ we have $s\theta_1 = t$ and
$s\theta_2 = t$ where $\theta_1 = \SET{F \mapsto z.\m{f}(\m{c})}$
and $\theta_2 = \SET{F \mapsto z.\m{f}(z)}$.
\end{example}

If the arguments of free variables contain at least one bound
variable, then, due to capture-avoidance, a matching substitution is
forced to use one of these arguments in order to construct terms with
bound variables. Furthermore, bound variables can only be used in the
contexts specified by the arguments of the free variable.

\begin{example}
Consider the term $s = x,y.F(\m{f}(x))$. For
$t = x,y.\m{g}(\m{f}(x))$ there is a unique substitution
$\theta = \SET{F \mapsto z.\m{g}(z)}$ such that $s\theta =
t$. Note that $s$ does not match $x,y.\m{g}(x)$ because the
bound variable $x$ does not occur in the context specified by the
argument of $F$ in $s$. Moreover, $s$ does not match $x,y.y$
because the bound variable $y$ does not occur in the arguments of
$F$.
\end{example}

Finally, if a free variable has multiple arguments, we have to ensure
that they are different enough such that there is always a unique
argument to choose while constructing a term during matching.

\begin{example}
Consider the term $s = x.F(\m{f}(x,\m{c}),z.\m{f}(x,z))$.
For $t = x.\m{f}(x,\m{c})$ we have $s\theta_1 = t$ and
$s\theta_2 = t$ where $\theta_1 = \SEt{F \mapsto z_1,z_2.z_1}$
and $\theta_2 = \SET{F \mapsto z_1,z_2.z_2(\m{c})}$.
\end{example}

Hence, DHPs as introduced in \cite{YHT04} are defined via
$\beta\eta$-normal forms and it is demanded that no argument of a free
variable is a higher-order subterm of another argument of the same
free variable. In the following, we present a definition of DHPs for
our definition of terms, i.e., $\lambda$-terms in $\beta\eta$-\emph{long}
normal from. To that end, we introduce the notion of expanded terms as
well as the expanded subterm relation which replaces the higher-order
subterm relation on $\beta\eta$-normal forms.
Intuitively, if $\vec{x}.s$ is expanded, then the $\eta$-normal from
of $s$ is not an abstraction.

\begin{definition}
A term $\vec{x}.s$ is \emph{expanded} if
$\vec{x}.s = \vec{x},\vec[k]{y}.h(\vec[m]{s},\vec[k]{y}\cf)$ such that
$\bigl(\fv{\vec[m]{s}} \cup \SET{h}\bigr) \cap \SEt{\vec[k]{y}} =
\varnothing$. Consider a term $\vec{x}.s$ as well as an expanded
term $\vec{x}.t = \vec{x},\vec[k]{y}.h(\vec[m]{t},\vec[k]{y}\cf)$.
We say that $\vec[n]{x}.t$ is an \emph{expanded subterm} of $\vec{x}.s$
$(\vec{x}.s \subtEeq \vec{x}.t)$ if there exist $n' \geq n$ and $k$
terms written as
$\vec[n']{x}.t_{m+1},\dots,\vec[n']{x}.t_{m+k}$ such that
$\vec{x}.s \subteq \vec[n']{x}.h(\vec[m+k]{t})$. We write
$\vec{x}.s \subtE \vec{x}.t$ if $\vec{x}.s \subtEeq \vec{x}.t$ and
$\vec{x}.s \neq \vec{x}.t$.
\end{definition}

\begin{example}
Consider the sort $\m{a}$ as well as the function symbols $\m{c} : \m{a}$,
$\m{g} : (\m{a},\m{a}) \to \m{a}$ and
$\m{f} : (\m{a},\m{a},\m{a} \to \m{a}) \to \m{a}$.
The terms $s_1 = \m{c}$, $s_2 = \m{g}(x,y)$, $s_3 = x.\m{g}(\m{c},x)$ and
$x.s_4 = x,y.\m{f}(x,\m{c},z.y(z))$ are expanded
but $s_5 = x.\m{g}(\m{c},y)$ and $x.s_6 = x,y.\m{f}(z.y(z),\m{c},z.y(z))$
are not.
Consider $x.u = x.\m{f}(x,\m{c},y.\m{g}(x,y))$. We have
$x.u \subtEeq x.v_i$ for $1 \leq i \leq 3$ where
$x.v_1 = x,z_1,z_2,z_3.\m{f}(z_1,z_2,z_4.z_3(z_4))$,
$x.v_2 = x,z.\m{f}(x,\m{c},y.z(y))$ and
$x.v_3 = x,z.\m{g}(x,z)$.
On the other hand, $x.u \nsubtEeq x.v_4$ where
$x.v_4 = x,y.\m{f}(\m{c},x,z.y(z))$.
\end{example}

In general, if $\vec{x}.t$ is an expanded term and
$s \subteq \vec{x}.t$ then also $s \subtEeq \vec{x}.t$.
Note that in the definition of $\subtEeq$ we do
not have to demand that $x_{n+1},\dots,x_{n'}$ are fresh
as these variables occur as bound variables in $\vec{x}.s$
and we follow the convention that no variable occurs both
free and bound in a term. Furthermore,
$\vec{x}.s \subt \vec{x}.t$ and
$\vec{x}.t \subtEeq \vec{x}.u$ implies $\vec{x}.s \subtE \vec{x}.u$.

\begin{definition}[adapted from \cite{YHT04}]
\label{def:dhp}
A term $s$ is a \emph{deterministic higher-order pattern} \textup{(DHP)} if
the following holds for all subterms $\vec{x}.y(\vec[m]{t})$ with
$y \notin \SEt{\vec{x}}$ and $1 \leq i \leq m$:
\begin{enumerate}[(i)]
\item \label{enu:dhp-var}
$\varnothing \neq \fv{t_i} \subseteq \SEt{\vec{x}}$,
\item \label{enu:dhp-exp}
$\vec{x}.t_i$ is an expanded term,
\item \label{enu:dhp-subt}
$\vec{x}.t_i \nsubtEeq \vec{x}.t_j$ whenever $i \neq j$.
\end{enumerate}
In this case, we say that $\vec[m]{\vec{x}.s}$ is a
\emph{DHP var-arg list}. Furthermore, we define
$\dhp \subseteq \trm{\xS}$ as the set of all \textup{DHPs}.
\end{definition}

\begin{example}
Consider the sort $\m{a}$, the function symbols $\m{f} : \m{a} \to \m{a}$,
$\m{c} : \m{a}$ as well as the variables $F : (\m{a},\m{a}) \to \m{a}$ and
$G : (\m{a},\m{a} \to \m{a}) \to \m{a}$.
The terms $x,y.F(y,x)$, $x,y.F(\m{f}(x),\m{f}(y))$ and
$x,y.F(x(y),x(\m{c}))$ are DHPs. Note that only the first term also
a pattern in the sense of Miller. On the other hand,
none of the terms $x,y.F(\m{c},x)$, $x,y.G(y,z.x(z,\m{c}))$,
$x,y.F(x,x)$, $x,y.F(\m{f}(x),x)$, and $x,y.G(x(y),z.x(z))$ is a
DHP because the first one violates condition \enuref{dhp-var}, the
second one violates condition \enuref{dhp-exp}
and the remaining terms violate condition \enuref{dhp-subt}.
\end{example}

We now show that DHPs have unique matching substitutions.

\begin{lemma}
\label{lem:matchUnique}
Let $s = \vec{x}.F(\vec[m]{s}) \in \dhp$ and $t$ be an arbitrary term. If
$s\theta = t$ where $F \in \dom{\theta}$ then $\theta(F)$ is unique.
\end{lemma}

\begin{proof}
Assume $s\theta = t$ and $\theta(F) = \vec[m]{y}.u$.
Without loss of generality, $\SEt{\vec{x}}$ is fresh for $\theta$.
Since $s \in \dhp$, $\vec{x}.s_i\theta = \vec{x}.s_i$ for
$1 \leq i \leq m$ by \defref{dhp}\enuref{dhp-var}.
Hence, $t = \vec{x}.u\mu$ where
$\mu = \SEt{\vec[m]{y} \mapsto \vec[m]{s}}$. We prove that
all subterms of $\vec[m]{y}.u$ are unique by induction on $\vec[m]{y}.u$.
Consider $\vec[m]{y}.u \subteq \vec[m]{y}.u' =
\vec[m]{y},\vec[k]{z}.h_1(\vec[l_1]{u})$, so
by $t = \vec{x}.u\mu$ there exists a subterm
$t \subteq \vec{x}.t' = \vec{x},\vec[k]{z}.h_2(\vec[l_2]{t})$
such that $\vec{x}.t' = \vec{x}.u'\mu$ (\textasteriskcentered).
We proceed by a case analysis on whether there exists some
$1 \leq j \leq m$ such that
$\vec{x}.s_j = \vec{x},\vec[l_4]{w}.h_2(\vec[l_3]{t},\vec[l_4]{w}\cf)$
where $l_2 = l_3 + l_4$.
\begin{itemize}
\item
If there is such a $j$, it is unique by \defref{dhp}\enuref{dhp-subt} and
from \defref{dhp}\enuref{dhp-exp} we know that
$(\fv{\vec[l_3]{t}} \cup \SEt{h_2}) \cap \SEt{\vec[l_4]{w}} = \varnothing$.
Furthermore, from \defref{dhp}\enuref{dhp-var} we obtain
$\varnothing \neq \fv{s_j} \subseteq \SEt{\vec{x}}$. Together with
(\textasteriskcentered), this forces $h_1 = y_j$ and therefore
$l_1 = l_4$.
Moreover,
$\vec{x},\vec[k]{z}.u_i\mu = \vec{x},\vec[k]{z}.t_{l_3+i}$ for
$1 \leq i \leq l_1$. By the induction hypothesis, these subterms are
unique, and so is $h = y_j$.
\item 
Otherwise, $h_1 = h_2$ and the claim follows immediately from the
induction hypothesis. \qed
\end{itemize}
\end{proof}

\begin{theorem}
\label{thm:matchUnique}
Let $s \in \dhp$ and $t$ be an arbitrary term.
If $s\theta = t$ then $\theta|_{\fv{s}}$ is unique.
\end{theorem}

\begin{proof}
Assume $s\theta = t$ and let $s = \vec{x}.h(\vec[m]{s})$. We prove that
$\theta|_{\fv{s}}$ is unique by induction on $s$. Without loss of
generality, $\SEt{\vec{x}}$ is fresh for $\theta$. If
$h \notin \dom{\theta}$ then
$s\theta = \vec{x}.h(\vec[m]{s}\theta)$ and therefore
$t = \vec{x}.h(\vec[m]{t})$ where $\vec{x}.s_i\theta = \vec{x}.t_i$ for
$1 \leq i \leq m$. By the induction hypothesis, the substitutions
$\theta|_{\fv{\vec{x}.s_i}}$ for $1 \leq i \leq m$ are unique, so
also $\theta|_{\fv{s}}$ is unique.
If $h \in \dom{\theta}$ we conclude by \lemref{matchUnique}. \qed
\end{proof}

As we can see in condition \enuref{dhp-var} of \defref{dhp}, it is
sufficient for deterministic matching
to require that each argument of a free variable contains
at least one bound variable. For the terms which can in principle be
used in FCU, it is additionally demanded that all leaves of an
argument term of a free variable are bound variables. In particular,
$z.F(\m{g}(\m{c},z))$ is a DHP but cannot be used in FCU as $\m{c}$ is not
a bound variable. Condition \enuref{dhp-exp} is also used in FCU. For
DHPs, the main motivation for this condition is that it facilitates a
more straightforward matching algorithm. Moreover, we use it in the
soundness and completeness proofs of our unification procedure.
Finally, condition \enuref{dhp-subt}, sometimes also referred to as
the \emph{local restriction}, is augmented in FCU by the \emph{global
restriction} which is the extension of \enuref{dhp-subt} using
$\nsubtE$ instead of $\nsubtEeq$ to all combinations of free variable
arguments in the unification problem.

\section{Unification}
\label{sec:unification}

We already showed in \exaref{dhp} that unification of DHPs is not
unitary. Therefore, instead of using most general unifiers as in
higher-order pattern unification and FCU, we have to work with
complete sets of unifiers as in full higher-order unification
\cite{SG89}. However, unlike full higher-order unification, we will
present a unification procedure which can produce complete sets of
unifiers which are minimal in the sense that the produced
substitutions are orthogonal.

\begin{definition}
\label{def:csu}
Given a set $E$ of unordered pairs of DHPs, we denote its set of
unifiers by $\unifiers{E}$. A \emph{complete set of unifiers} $U$ of
$E$, which we abbreviate by $\csu{E}$, satisfies the following
conditions:
\begin{enumerate}[(i)]
\item \label{enu:csuDom}
$\dom{\theta} \subseteq \fv{E}$ for all $\theta \in U$,
\item \label{enu:csuUnif}
$U \subseteq \unifiers{E}$,
\item \label{enu:csuComplete}
for every $\delta \in \unifiers{E}$ there exists some
$\theta \in U$ such that $\theta \subsumes \delta\restr{\fv{E}}$,
\end{enumerate}
A \emph{minimal complete set of unifiers} $U$ of $E$
\textup{($\mcsu{E}$)} is a $\csu{E}$ with the additional property that
$\theta_1, \theta_2 \in U$ and $\theta_1 \neq \theta_2$ implies
$\theta_1 \perp \theta_2$.
\end{definition}

The following example shows that minimal complete sets of unifiers
of DHPs can be infinite.

\begin{example}
\label{exa:infinite}
Consider the sort $\m{a}$, the function symbol
$\m{f} : \m{a} \to \m{a}$ as well as the variable
$M : \m{a} \to \m{a}$. The terms $x.M(\m{f}(x))$ and $x.\m{f}(M(x))$
admit the infinite minimal complete set of unifiers
$\SET{\theta_k \mid k \leq 0}$ where
$\theta_k = \SET{M \mapsto x.\m{f}^k(x)}$.
\end{example}

We proceed with the formal definition of unification of DHPs. As in
full higher-order unification, an important ingredient is the notion
of partial bindings. Due to the nature of DHPs, we do not need to
consider partial bindings where the head is a free variable.

\begin{definition}
A \emph{partial binding} of type $\vec{\sigma} \to a$ is a term $u$
of the form
\[
\vec{x}.h(\vec[k_1]{y^1}.H_1(\vec{x}\cf,\vec[k_1]{y^1}\cf),\dots,
\vec[k_m]{y^m}.H_m(\vec{x}\cf,\vec[k_m]{y^m}\cf)) \in \dhp
\]
for some $h \in \xF \cup \SEt{\vec{x}}$ where
$h : \vec[m]{\tau} \to a$, $\vec{x} : \vec{\sigma}$,
$\tau_i = \vec[k_i]{\varphi^i} \to b_i$,
$\vec[k_i]{y^i} : \vec[k_i]{\varphi^i}$ and
$H_i : (\vec{\sigma},\vec[k_i]{\varphi^i}) \to b_i$ for all
$1 \leq i \leq m$. Note that partial bindings are uniquely
determined by their type and the choice of $h$ up to renaming of the
variables $H_1,\dots,H_m$ which we always assume to be fresh. If
$h = x_i$ for some $1 \leq i \leq n$ then the partial binding is
called a \emph{projection binding}, otherwise it is an
\emph{imitation binding}. If $u$ is a(n) projection binding
(imitation binding) of type $\tau$, we write $u \in \pb{\tau}$
\textup{($u \in \ib{\tau}$)}.
\end{definition}

\begin{definition}
\label{def:unif}
The following inference rules enable the nondeterministic computation
of complete sets of unifiers for deterministic higher-order patterns.
To that end, they operate on \emph{unification pairs} $(E,\theta)$
where $E$ is a finite set of unordered pairs of DHPs of the same type
and $\theta$ is a substitution.
\begin{itemize}
\item[\urem] \emph{removal of trivial equations}
\[
\frac{\SET{s \approx s} \uplus E, \theta}{E, \theta}
\]
\item[\udec] \emph{decomposition}
\[
\frac{\SET{\vec[k]{x}.h(\vec{s}) \approx \vec[k]{x}.h(\vec{t})} \uplus E,
\theta}{\bigcup_{1 \leq i \leq n}
\SEt{\vec[k]{x}.s_i \approx \vec[k]{x}.t_i} \cup E, \theta}
\tag*{if $h \in \xF \cup \SEt{\vec[k]{x}}$}
\]
\item[\uvar] \emph{variable elimination}
\[
\frac{\SET{\vec[k]{x}.F(\vec[k]{x}) \approx u} \uplus E, \theta}
{E\SET{F \mapsto u}, \theta\SET{F \mapsto u}}
\tag*{if $F \notin \fv{u}$}
\] 
\item[\uimt] \emph{imitation}
\[
\frac{E = \SET{\vec[k]{x}.F(\vec{s}) \approx \vec[k]{x}.h(\vec[m]{t})}
\uplus E', \theta}{E\SET{F \mapsto u}, \theta\SET{F \mapsto u}}
\tag*{if $F : \tau$ and $u \in \ib{\tau}$}
\]
where $\hd{u} = h \in \xF$
\item[\uprj] \emph{projection}
\[
\frac{E = \SET{\vec[k]{x}.F(\vec{s}) \approx \vec[k]{x}.h(\vec[m]{t})}
\uplus E', \theta}{E\SET{F \mapsto u}, \theta\SET{F \mapsto u}}
\tag*{if $F : \tau$ and $u \in \pb{\tau}$}
\]
where $h = \hd{(\vec[k]{x}.F(\vec{s}))\SET{F \mapsto u}} \in
\xF \cup \SEt{\vec[k]{x}}$
\item[\uffe] \emph{same variables}
\[
\frac{\SET{\vec[k]{x}.F(\vec{s}) \approx \vec[k]{x}.F(\vec{t})} \uplus E,
\theta}{E\SET{F \mapsto u}, \theta\SET{F \mapsto u}}
\tag*{if $u = \vec{y}.H(y_{i_1}\cf,\dots,y_{i_m}\cf)$}
\]
where $H$ is a fresh variable and $\SEt{\seq[m]{i}} =
\SEt{1 \leq i \leq n \mid \vec[k]{x}.s_i = \vec[k]{x}.t_i}$
\item[\uffne] \label{enu:ffneq} \emph{different variables}
\[
\frac{\SET{\vec[k]{x}.F(\vec{s}) \approx \vec[k]{x}.G(\vec[m]{t})}
\uplus E, \theta}{E\SET{F \mapsto u, G \mapsto v},
\theta\SET{F \mapsto u, G \mapsto v}}
\tag*{if $F \neq G$}
\]
where $u = \vec{y}.H(\vec[l]{u})$, $v = \vec[m]{z}.H(\vec[l]{v})$,
$H$ is a fresh variable,
$\SET{(\vec{y}.u_i,\vec[m]{z}.v_i) \mid 1 \leq i \leq l} = P_1 \cup P_2$
for
\begin{align*}
P_1 &= \SET{(\vec{y}.y_i\cf,\theta_i(G_i)) \mid 1 \leq i \leq n,
\vec[k]{x}.s_i = (\vec[k]{x}.G_i(\vec[m]{t}))\theta_i} \\
P_2 &= \SET{(\theta_i(H_i),\vec[m]{z}.z_i\cf) \mid 1 \leq i \leq m,
\vec[k]{x}.t_i = (\vec[k]{x}.H_i(\vec{s}))\theta_i}
\end{align*}
\end{itemize}
We write $(E,\theta) \ustep (E',\theta')$ if there is a corresponding
one-step derivation in the inference system. A
derivation $(E,\varnothing) \ustep^* (\varnothing,\theta)$ is
considered to be \emph{successful}.
\end{definition}

The inference system is defined along the lines of the one given in
\cite{SG89} for full higher-order unification. A notable difference is
the handling of flex-flex pairs (pairs where the heads of both
terms are free variables) in \uffe and \uffne. While partial
bindings have to be used here for full higher-order unification, there
is actually a most general way of resolving them for DHPs: In rule
\uffe we obtain a most general unifier of the pair by only keeping the
arguments which are equal while in rule \uffne we also have to keep
arguments which can be constructed by combining possibly multiple
arguments without using free variables. Hence, we can restrict \udec
to the cases where the head is not a free variable. Furthermore,
partial bindings do not need to consider free variables as
heads. Moreover, in \uprj, the resulting head symbol has to be equal
because in DHPs the arguments of free variables cannot be projections.
Finally, the inference
system is finitely branching as one can only apply an inference rule
to one of the finitely many unordered pairs, and each of them has a
deterministic (up to renaming of fresh variables) outcome except for
projections of which there are always only finitely many. The
following three examples showcase our inference system for
unification of DHPs. All of them
are out of scope for both Miller's pattern unification and FCU.

\begin{example}
Consider the sort $\m{a}$, function symbols
$\m{f} : \m{a} \to \m{a}$ and $\m{g} : (\m{a},\m{a}) \to \m{a}$ as
well as the variables $F : (\m{a},\m{a}) \to \m{a}$ and
$G : (\m{a} \to \m{a},\m{a},\m{a}) \to \m{a}$. The terms
$x,y.F(x,\m{g}(x,y))$ and $x,y.G(z.\m{g}(x,z),\m{f}(x),y)$ are
unifiable as witnessed by the following derivation
\begin{alignat*}{2}
&&~&
(\SET{x,y.F(x,\m{g}(x,y)) \approx x,y.G(z.\m{g}(x,z),\m{f}(x),y)},
\varnothing) \\
&\ustep_\uffne&&
(\SET{x,y.H(\m{g}(x,y),\m{f}(x)) \approx x,y.H(\m{g}(x,y),\m{f}(x))},
\theta) \\
&\ustep_\urem&& (\varnothing,\theta)
\end{alignat*}
where $H$ is a fresh variable of appropriate type and
\[
\theta = \SET{F \mapsto y_1,y_2.H(y_2,\m{f}(y_1)),
G \mapsto z_1,z_2,z_3.H(z_1(z_3),z_2)}
\]
\end{example}

\begin{example}
Consider the sort $\m{a}$, the function symbol
$\m{f} : \m{a} \to \m{a}$ as well as the variables
$M,N : (\m{a},\m{a}) \to \m{a}$. Then,
$(\SET{x,y.M(\m{f}(x),\m{f}(y)) \approx x,y.\m{f}(N(y,x))},\varnothing)$
is a unification pair
representing the unification problem discussed in \exaref{dhp}.
From the initial unification pair, there are three possible steps
with $\ustep$. First, we can use \uprj to project $M$ to its first
argument and successfully complete the derivation to obtain unifier
(i) from \exaref{dhp}:
\begin{alignat*}{2}
&&~&
(\SET{x,y.M(\m{f}(x),\m{f}(y)) \approx x,y.\m{f}(N(y,x))}, \varnothing) \\
&\ustep_\uprj&&
(\SET{x,y.\m{f}(x) \approx x,y.\m{f}(N(y,x))},
\SEt{M \mapsto z_1,z_2.z_1}) \\
&\ustep_\udec&&
(\SET{x,y.x \approx x,y.N(y,x)},\SEt{M \mapsto z_1,z_2.z_1}) \\
&\ustep_\uprj&& (\SET{x,y.x \approx x,y.x},
\SEt{M \mapsto z_1,z_2.z_1, N \mapsto z_1,z_2.z_2}) \\
&\ustep_\urem&&
(\varnothing, \SEt{M \mapsto z_1,z_2.z_1, N \mapsto z_1,z_2.z_2})
\end{alignat*}
Similarly, we can use \uprj to project $M$ to its second argument
and therefore obtain unifier (ii) from \exaref{dhp}:
\begin{alignat*}{2}
&&~&
(\SET{x,y.M(\m{f}(x),\m{f}(y)) \approx x,y.\m{f}(N(y,x))}, \varnothing) \\
&\ustep_\uprj&& (\SET{x,y.\m{f}(y) \approx x,y.\m{f}(N(y,x))},
\SEt{M \mapsto z_1,z_2.z_2}) \\
&\ustep^*&&
(\varnothing, \SEt{M \mapsto z_1,z_2.z_2, N \mapsto z_1,z_2.z_1})
\end{alignat*}
Finally, we can use \uimt and eventually obtain unifier (iii) from
\exaref{dhp}:
\begin{alignat*}{2}
&&~&
(\SET{x,y.M(\m{f}(x),\m{f}(y)) \approx x,y.\m{f}(N(y,x))}, \varnothing) \\
&\ustep_\uimt&&
(\SET{x,y.\m{f}(H(\m{f}(x),\m{f}(y))) \approx x,y.\m{f}(N(y,x))},
\theta_1) \\
&\ustep_\udec&& (\SET{x,y.H(\m{f}(x),\m{f}(y)) \approx x,y.N(y,x)},
\theta_1) \\
&\ustep_\uffne&&
(\SET{x,y.Z(\m{f}(x),\m{f}(y)) \approx x,y.Z(\m{f}(x),\m{f}(y))},
\theta) \\
&\ustep_\urem&& (\varnothing,\theta)
\end{alignat*}
where $H$, $Z$ are fresh variables,
$\theta_1 = \SET{M \mapsto z_1,z_2.\m{f}(H(z_1,z_2))}$ and
\[
\theta = \SET{M \mapsto z_1,z_2.\m{f}(Z(z_1,z_2)),
H \mapsto z_1,z_2.Z(z_1,z_2),
N \mapsto z_1,z_2.Z(\m{f}(z_2),\m{f}(z_1))}
\]
Note that $\theta \subsumes \delta\restr{\SET{M,N}}$ where
$\delta$ is unifier (iii) from \exaref{dhp}.
\end{example}

\begin{example}
\label{exa:infinite2}
Recall \exaref{infinite}. Starting from
$(\SET{x.M(\m{f}(x)) \approx x.\m{f}(M(x))},\varnothing)$, we can
apply \uprj to get
$(\SET{x.\m{f}(x) \approx x.\m{f}(x)},\SET{M \mapsto z.z})$ and
therefore one of the solutions by one application of \urem, or we
apply \uimt to get
$(\SET{x.\m{f}(H(\m{f}(x))) \approx x.\m{f}(\m{f}(H(x)))},
\SET{M \mapsto z.\m{f}(H(z))})$ where $H$ is a fresh variable. After
applying \udec we are back to a renamed variant of the initial
equation. Hence, although the inference system is finitely
branching, its nontermination facilitates the enumeration of
infinite set of unifiers.
\end{example}

The following result states that the inference system is well-defined,
i.e., applying a step does not yield terms which are not DHPs. Even
though the statement sounds trivial, the proof is quite technical since
we have to ensure that the local restriction is never violated.
Intuitively, the theorem holds because the substitutions which are applied
to the unification problem in the inference system map variables to DHPs,
and if such a substitution is applied to a DHP, the result is again a DHP.
A detailed proof of the theorem is available in \appref{dhpu}.

\begin{theorem}
\label{thm:dhpu}
If $(E,\theta) \ustep (E',\theta')$ and $E \subseteq \dhp \times \dhp$
then $E' \subseteq \dhp \times \dhp$.
\end{theorem}

\section{Correctness}
\label{sec:correct}

We start by proving soundness of the inference system, i.e.,
successful derivations produce unifiers. In the following,
$\theta \cdot U$ denotes $\SET{\theta\delta \mid \delta \in U}$.

\begin{lemma}
\label{lem:sound}
If $(E_1,\theta) \ustep (E_2,\theta\delta)$ then
$\delta \cdot \unifiers{E_2} \subseteq \unifiers{E_1}$.
\end{lemma}

\begin{proof}
We proceed by case analysis on the inference rule applied in the
step. For \urem and \udec, the result is trivial because $\delta$
is the identity substitution and $\unifiers{E_2} = \unifiers{E_1}$.
Regarding \uvar, we have
$E_1 = \SET{\vec[k]{x}.F(\vec[k]{x}) \approx u} \uplus E$ with
$F \notin \fv{u}$, $\delta = \SET{F \mapsto u}$ and $E_2 = E\delta$.
Let $\gamma \in \unifiers{E_2}$, so
$\delta\gamma \in \unifiers{E}$. Since $\delta\gamma$ also unifies
$\vec[k]{x}.F(\vec[k]{x})$ and $u$,
$\delta\gamma \in \unifiers{E_1}$. For \uimt and \uprj, we have
$E_2 = E_1\delta$, so if $\gamma \in \unifiers{E_2}$ then
$\delta\gamma \in \unifiers{E_1}$ as desired. Next, consider
\uffe. We have
$E_1 = \SET{\vec[k]{x}.F(\vec{s}) \approx \vec[k]{x}.F(\vec{t})}
\uplus E$, $\delta = \SET{F \mapsto u}$ where $u$ is defined as in
\defref{unif} and $E_2 = E\delta$. Let $\gamma \in \unifiers{E_2}$,
so $\delta\gamma \in \unifiers{E}$. Since $\delta\gamma$ also
unifies $\vec[k]{x}.F(\vec{s})$ and $\vec[k]{x}.F(\vec{t})$,
$\delta\gamma \in \unifiers{E_1}$. Finally, consider \uffne. We have
$E_1 = \SET{\vec[k]{x}.F(\vec{s}) \approx \vec[k]{x}.G(\vec[m]{t})}
\uplus E$, $\delta = \SET{F \mapsto u, G \mapsto v}$ where $u$ and
$v$ are defined as in \defref{unif} and $E_2 = E\delta$. Let
$\gamma \in \unifiers{E_2}$, so $\delta\gamma \in \unifiers{E}$.
Without loss of generality,
$\SEt{\vec[k]{x}}$ is fresh for $\delta\gamma$, so
$(\vec[k]{x}.F(\vec{s}))\delta\gamma =
\vec[k]{x}.H(\vec[l]{u}\SEt{\vec{y} \mapsto \vec{s}}\delta\gamma)$ and
$(\vec[k]{x}.G(\vec[m]{t}))\delta\gamma =
\vec[k]{x}.H(\vec[l]{v}\SEt{\vec[m]{z} \mapsto \vec[m]{t}}\delta\gamma)$.
Furthermore, we have
$\vec[k]{x}.u_i\SEt{\vec{y} \mapsto \vec{s}} =
\vec[k]{x}.v_i\SEt{\vec[m]{z} \mapsto \vec[m]{t}}$ for all
$1 \leq i \leq l$ by definition. Hence, $\delta\gamma$ also unifies
$\vec[k]{x}.F(\vec{s})$ and $\vec[k]{x}.G(\vec[m]{t})$ and
$\delta\gamma \in \unifiers{E_1}$ as desired. \qed
\end{proof}

\begin{theorem}
\label{thm:sound}
If $(E,\varnothing) \ustep^* (\varnothing,\theta)$ then
$\theta|_{\fv{E}} \in \unifiers{E}$.
\end{theorem}

\begin{proof}
We start by proving that $(E_1,\varnothing) \ustep^n (E_2,\theta_2)$
implies $\theta_2 \cdot \unifiers{E_2} \subseteq \unifiers{E_1}$ for
all $n \geq 0$ by induction on $n$. The base case $n = 0$ is trivial.
For the step case, consider $n > 0$, so
$(E_1,\theta) \ustep^{n-1} (E_3,\theta_3) \ustep (E_2,\theta_2)$
where $\theta_2 = \theta_3\delta$ for some substitution $\delta$.
From \lemref{sound} we obtain
$\delta \cdot \unifiers{E_2} \subseteq \unifiers{E_3}$.
Furthermore, the induction hypothesis yields
$\theta_3 \cdot \unifiers{E_3} \subseteq \unifiers{E_1}$. Since
substitution composition is associative,
$\theta_2 \cdot \unifiers{E_2} = \theta_3 \cdot
(\delta \cdot \unifiers{E_2}) =  \subseteq \unifiers{E_1}$
which concludes the proof
of this auxiliary statement. In particular, for
$(E,\varnothing) \ustep^* (\varnothing,\theta)$ we obtain
$\theta \cdot \unifiers{\varnothing} \subseteq \unifiers{E}$ as a
special case. Since the identity substitution is an element of
$\unifiers{\varnothing}$ and the restriction has no effect on terms
in $E$, we conclude that $\theta|_{\fv{E}} \in \unifiers{E}$. \qed
\end{proof}

Next, we turn our attention to nondeterministic completeness:
If $\delta \in \unifiers{E}$ then there exists a successful
derivation starting with $E$ which produces a unifier
$\theta$ such that $\theta \subsumes \delta\restr{\fv{E}}$.
The upcoming definition and lemmata establish the existence of a
terminating relation $\custep_W$ which will be used to construct a
successful run from an existing unifier. The proof idea originates from
the respective proof for full higher-order unification in \cite{SG89}.

\begin{definition}
\label{def:custep}
Consider a unification pair $(E,\theta)$, a set of variables $W$ and
an idempotent substitution $\delta \in \unifiers{E}$ such that
$\theta\delta = \delta\restr{W}$. We write
$(E,\theta,\delta) \custep_W (E',\theta',\delta')$ if
$(E,\theta) \ustep (E',\theta')$, $\delta'$ is an idempotent
substitution such that $\delta' \in \unifiers{E'}$ and
$\theta'\delta' = \delta' = \delta\restr{W}$.
\end{definition}

\begin{lemma}
\label{lem:nondetComplete}
Consider a unification pair $(E,\theta)$, a set of variables $W$ and
an idempotent substitution $\delta \in \unifiers{E}$ such that
$\theta\delta = \delta\restr{W}$. If $E \neq \varnothing$ then
$(E,\theta,\delta) \custep_W (E',\theta',\delta')$ for some
$(E',\theta',\delta')$.
\end{lemma}

\begin{lemma}
\label{lem:custepTerm}
The relation $\custep_W$ is terminating for every set $W$ of variables.
\end{lemma}

The proofs of the preceding lemmata can be found in \appref{lemmata}.
We are now ready to establish nondeterministic completeness.

\begin{theorem}
\label{thm:nondetComplete}
Let $\delta \in \unifiers{E}$. There exists a derivation
$(E,\varnothing) \ustep^* (\varnothing,\theta)$ such that
$\theta \subsumes \delta\restr{\fv{E}}$.
\end{theorem}

\begin{proof}
There exists an idempotent variant $\delta'$ of $\delta$ such that
$\delta' \subsumes \delta\restr{\fv{E}}$. Since
$\delta' \in \unifiers{E}$ and
$\varnothing\delta' = \delta'\restr{\fv{E}}$, by repeated
applications of \lemref{nondetComplete} as well as termination of
$\custep_{\fv{E}}$ (\lemref{custepTerm}) there must be a maximal
sequence of transformations
\[
(E,\varnothing,\delta') \custep_{\fv{E}} (E_1,\theta_1,\delta_1)
\custep_{\fv{E}} \cdots \custep_{\fv{E}} (E_k,\theta_k,\delta_k) =
(\varnothing,\theta_k,\delta_k)
\]
where the final equality follows from \lemref{nondetComplete}.
From the definition of $\custep_{\fv{E}}$ we obtain
$(E,\varnothing) \ustep^* (\varnothing,\theta_k)$,
$\delta' = \delta_k\restr{\fv{E}}$,
$\theta_k \subsumes \delta_k\restr{\fv{E}}$ and therefore also
$\theta_k \subsumes \delta' \subsumes \delta \restr{\fv{E}}$ as
desired. \qed
\end{proof}

From soundness and nondeterministic completeness we infer that
the substitutions produced by our unification procedure
constitute a complete set of unifiers.

\begin{theorem}
\label{thm:complete}
For any set of unordered pairs of \textup{DHPs} of equal type $E$ with
$\unifiers{E} \neq \varnothing$, The set
$\SET{\theta|_{\fv{E}} \mid (E,\varnothing) \ustep^* (\varnothing,\theta)}$
is a $\csu{E}$.
\end{theorem}

\begin{proof}
Let $U = \SET{\theta|_{\fv{E}} \mid (E,\varnothing) \ustep^*
(\varnothing,\theta)}$. By \thmref{sound},
$U \subseteq \unifiers{E}$. Furthermore, \thmref{nondetComplete}
establishes that for every $\delta \in \unifiers{E}$ there exists
some $\theta|_{\fv{E}} \in U$ such that
$\theta \subsumes \delta \restr{\fv{E}}$ and thus also
$\theta|_{\fv{E}} \subsumes \delta \restr{\fv{E}}$. Hence, $U$ is a
$\csu{E}$ according to \defref{csu}. \qed
\end{proof}

\section{Minimality}
\label{sec:minimal}

In this section, we show that implementations of the unification
inference system which adhere to the mild restrictions given in the
following definition are guaranteed to produce minimal complete
sets of unifiers if successful.

\begin{definition}
A \emph{strategy} $\fS$ is a mapping from unification pairs
$(E,\theta)$ with $E \neq \varnothing$ to a particular element
$\fS((E,\theta)) = u \approx v \in E$ on which an inference rule
should be applied. In particular,
$(E,\theta) \ustep_{\fS} (E',\theta')$ denotes a step in the
inference system restricted to $\fS$. A derivation with
$\ustep_{\fS}$ is called \emph{locally optimal} if \urem and \uvar
are given precedence to all other inference rules. Finally, a
derivation with $\ustep_{\fS}$ is called
\emph{fresh-var-deterministic} if fresh variable names are computed
in a deterministic way.
\end{definition}

\begin{lemma}
\label{lem:diffStepIncomparableUnif}
Let $\fS$ be a strategy. If $(E,\theta) \ustep_{\fS} (E_1,\theta_1)$
and $(E,\theta) \ustep_{\fS} (E_2,\theta_2)$ are locally optimal and
fresh-var-deterministic and $\theta_1 \neq \theta_2$ then
$\theta_1 \perp \theta_2$.
\end{lemma}

\begin{proof}
Let $(E,\theta) \ustep_{\fS} (E_1,\theta_1)$,
$(E,\theta) \ustep_{\fS} (E_2,\theta_2)$ such that
$\theta_1 \neq \theta_2$ and $\fS((E,\theta)) = s \approx t$ where
$s = \vec[k]{x}.h_1(\vec{s})$ and $t = \vec[k]{x}.h_2(\vec[m]{t})$.
If $s = t$ then both steps are the same applications of \urem which
contradicts our assumption $\theta_1 \neq \theta_2$. Hence,
$s \neq t$. If $s = \vec[k]{x}.F(\vec[k]{x}) \approx v$ and
$F \notin \fv{v}$ then both steps are the same applications of
\uvar which again contradicts our assumption
$\theta_1 \neq \theta_2$. If $h_1 = h_2$ then at most one of \udec
or \uffe is applicable, and both have a deterministic outcome
because the derivations are fresh-var-deterministic. Therefore, this
case is impossible by our assumption $\theta_1 \neq \theta_2$.
Hence, $h_1 \neq h_2$. If
$h_1, h_2 \in \xV \setminus \SEt{\vec[k]{x}}$ then only \uffne is
applicable which has a deterministic outcome because the derivations
are fresh-var-deterministic. Therefore, this case is also impossible
by our assumption $\theta_1 \neq \theta_2$. Since there is no rule
for $h_1 \neq h_2$ where $h_1, h_2 \in \xF \cup \SEt{\vec[k]{x}}$,
without loss of generality,
$h_1 : \tau \in \xV \setminus \SEt{\vec[k]{x}}$ and
$h_2 \in \xF \cup \SEt{\vec[k]{x}}$. In this case, \uimt or \uprj
are applicable. Hence, $\theta_1 = \theta\SET{F \mapsto \vec{y}.u}$
and $\theta_2 = \theta\SET{F \mapsto \vec{y}.v}$ such that
$\vec{y}.u,\vec{y}.v \in \ib{\tau} \cup \pb{\tau}$ and $u \neq v$.
By definition of $\ib{\tau} \cup \pb{\tau}$ and the
assumption that the derivations are fresh-var deterministic,
$\hd{\vec{y}.u} \neq \hd{\vec{y}.v}$ and
$\hd{\vec{y}.u},\hd{\vec{y}.v} \in \xF \cup \SEt{\vec{y}}$, so there
do not exist substitutions $\mu_1$, $\mu_2$ such that
$(\vec{y}.u)\mu_1 = (\vec{y}.v)\mu_2$. Thus, there does not exist a
substitution $\gamma$ such that $\theta_1 \subsumes \gamma$ and
$\theta_2 \subsumes \gamma$, so $\theta_1 \perp \theta_2$ as
desired. \qed
\end{proof}

The following lemma is an easy observation about the inference system
which follows by a straightforward induction argument together with
transitivity of $\subsumes$.

\begin{lemma}
\label{lem:ustepSubstLE}
If $(E,\theta) \ustep^* (E',\theta')$ then $\theta \subsumes \theta'$.
\end{lemma}

\begin{theorem}
\label{thm:diffRunIncomparableUnif}
Let $\fS$ be a strategy. If
$(E,\varnothing) \ustep_{\fS}^* (\varnothing,\theta_1)$ and
$(E,\varnothing) \ustep_{\fS}^* (\varnothing,\theta_2)$ are locally
optimal and fresh-var-deterministic and $\theta_1 \neq \theta_2$ then
$\theta_1 \perp \theta_2$.
\end{theorem}

\begin{proof}
Since $\theta_1 \neq \theta_2$, the two derivations are not empty and
there must be a first point of divergence, i.e.,
$(E,\varnothing) \ustep_{\fS}^* (E',\theta')$,
$(E',\theta') \ustep_{\fS} (E_1',\theta_1')$,
$(E',\theta') \ustep_{\fS} (E_2',\theta_2')$,
$(E_1',\theta_1') \ustep_{\fS}^* (\varnothing,\theta_1)$ and
$(E_2',\theta_2') \ustep_{\fS}^* (\varnothing,\theta_2)$
such that $\theta_1' \neq \theta_2'$. From
\lemref{diffStepIncomparableUnif} we obtain
$\theta_1' \perp \theta_2'$. Furthermore, \lemref{ustepSubstLE}
yields $\theta_1' \subsumes \theta_1$ as well as
$\theta_2' \subsumes \theta_2$. If there is a substitution $\gamma$ such
that $\theta_1 \subsumes \gamma$ and $\theta_2 \subsumes \gamma$ then
by transitivity of $\subsumes$ also $\theta_1' \subsumes \gamma$ and
$\theta_2' \subsumes \gamma$ which contradicts $\theta_1' \perp \theta_2'$.
Hence, $\theta_1 \perp \theta_2$ as desired. \qed
\end{proof}

\begin{theorem}
\label{thm:mincomplete}
For any set of unordered pairs of \textup{DHPs} of equal type $E$ with
$\unifiers{E} \neq \varnothing$ and strategies $\fS$, The set
$\SET{\theta|_{\fv{E}} \mid (E,\varnothing) \ustep_{\fS}^*
(\varnothing,\theta)}$ restricted to locally optimal and
fresh-var-deterministic derivations is a $\mcsu{E}$.
\end{theorem}

\begin{proof}
Let $U = \SET{\theta|_{\fv{E}} \mid (E,\varnothing) \ustep_{\fS}^*
(\varnothing,\theta)}$. Since \lemref{nondetComplete} and
therefore \thmref{nondetComplete} are strategy-agnostic, produce
locally optimal derivations and do not depend on the names of fresh
variables, \thmref{complete} yields that $U$ is a $\csu{E}$.
From \thmref{diffRunIncomparableUnif} we obtain that for all
$\theta_1, \theta_2 \in U$ with $\theta_1 \neq \theta_2$ we have
$\theta_1 \perp \theta_2$. Hence, $U$ is a $\mcsu{E}$ according to
\defref{csu}. \qed
\end{proof}

A prototype implementation of the inference system which follows a
fixed strategy and only produces locally optimal and
fresh-var-deterministic derivations is available on GitHub.%
\footnote{\url{https://github.com/niedjoh/hrstk}}
Regarding nontermination of the inference system
as shown in \exaref{infinite2}, the strategy employed by our
prototype implementation postpones processing of equations of the form
$\vec[k]{x}.F(\vec{s}) \approx t$ where $F \in \fv{t}$ as long as
possible and simply gives up whenever all remaining equations
are of this form in order to avoid infinite behavior.
Note that due to our usage of a $\beta\eta$-free formulation of types,
our prototype can do without an implementation of $\beta$-reduction.

\section{Conclusion}
\label{sec:conclusion}

We introduced a sound and complete inference system to compute
complete sets of unifiers of higher-order unification problems
containing DHPs. Furthermore, we showed that suitable restrictions of
implementations guarantee minimal complete sets of unifiers if
successful. Our inference system is closely related to the one given
in \cite{SG89} for full higher-order unification, but with notable
differences such as most general solutions to flex-flex
pairs and predictable head symbols in projections.
Furthermore, DHPs are closely related to the terms used in FCU
which is decidable and unitary. However, these nice properties of FCU
are achieved by a global restriction on the unification problem, so it
can be hard to predict when it is actually usable as a subroutine in
practice. While FCU can solve unification goals directly, in
unification of DHPs one has to resort to partial bindings as in full
higher-order unification to keep the problem
manageable. Nevertheless, FCU can be simulated by our approach: If
there exists a most general unifier of a FCU problem, there is a
corresponding fresh-var deterministic and finite derivation in our
inference system by \thmref{complete}. Since the result is a most
general unifier of the problem, \thmref{mincomplete} ensures that
this is the only possible derivation.

The existence of infinite minimal complete sets
of unifiers renders a termination proof of our derivation relation
impossible and leaves decidability of the unification problem as a
nontrivial open
question. Note that an exact way of detecting situations which require an
infinite minimal complete set of unifiers would not be enough since
the inference system would still be infinitely branching without a
clever representation of infinite sets of unifiers. Since it
is not obvious that these situations only arise at the end of the
unification process, every rule in the inference system potentially
has to deal with abstract representation of infinite sets of terms.
The idea to use
(regular) grammars as stated in the conclusion of \cite{VBN21} seems to
be a promising step in that direction.

\begin{credits}
\subsubsection{\ackname}
We thank the anonymous reviewers for their valuable comments
which improved the presentation of the paper.

\subsubsection{\discintname}
The authors have no competing interests to declare that are
relevant to the content of this article.
\end{credits}

\bibliographystyle{splncs04}
\bibliography{references}

\appendix

\section{Proof of \thmref{dhpu}}
\label{app:dhpu}

We start by proving transitivity of $\subtEeq$ for a special case.

\begin{lemma}
\label{lem:subtEtrans}
If $\vec{x}.s \subtEeq \vec{x}.t$, $\vec{x}.t \subtEeq \vec{x}.u$ and
$\varnothing \neq \fv{u} \subseteq \SEt{\vec{x}}$
then $\vec{x}.s \subtEeq \vec{x}.u$.
\end{lemma}

\begin{proof}
Let $\vec{x}.t = \vec{x},\vec[k_1]{y}.h_1(\vec[m_1]{t},\vec[k_1]{y}\cf)$
and $\vec{x}.u = \vec{x},\vec[k_2]{y}.h_2(\vec[m_2]{u},\vec[k_2]{y}\cf)$.
By definition of $\subtEeq$, there are $n_1 \geq n$ and terms
$\vec[n_1]{x}.t_{m_1+1},\dots,\vec[n_1]{x}.t_{m_1+k_1}$ such that
$\vec{x}.s \subteq \vec[n_1]{x}.h_1(\vec[m_1+k_1]{t})$.
Furthermore, there is some $n_2$ and terms
\[
\vec{x},\vec[k_1]{y},\vec[n_2]{z}.u_{m_2+1},\dots,\vec{x},
\vec[k_1]{y},\vec[n_2]{z}.u_{m_2+k_2}
\]
such that
$\vec{x}.t \subteq
\vec{x},\vec[k_1]{y},\vec[n_2]{z}.h_2(\vec[m_2+k_2]{u})$. If
$\vec{x}.t = \vec{x},\vec[k_1]{y},\vec[n_2]{z}.h_2(\vec[m_2+k_2]{u})$,
we are done immediately. Otherwise, by
$\varnothing \neq \fv{u} \subseteq \SEt{\vec{x}}$, we know that
$\vec{x},\vec[k_1]{y}.t_j \subteq
\vec{x},\vec[k_1]{y},\vec[n_2]{z}.h_2(\vec[m_2+k_2]{u})$ for some
$1 \leq j \leq m_1$. Since $\vec{x}.t$ is an expanded term, we can
remove $\vec[k_1]{y}$ to obtain
$\vec{x}.t_j \subteq
\vec{x},\vec[n_2]{z}.h_2(\vec[m_2+k_2]{u})$. Hence, also
$\vec[n_1]{x}.h_1(\vec[m_1+k_1]{t}) \subteq
\vec[n_1]{x},\vec[n_2]{z}.h_2(\vec[m_2+k_2]{u})$. By transitivity
of $\subteq$, we conclude that
$\vec{x}.s \subteq \vec[n_1]{x},\vec[n_2]{z}.h_2(\vec[m_2+k_2]{u})$
and thus $\vec{x}.s \subtEeq \vec{x}.u$. \qed
\end{proof}

Next we show that $\subtEeq$ enjoys a special kind of closedness
under substitutions whose image consists expanded terms.

\begin{lemma}
\label{lem:subtEclosedUnderSubst}
Let $\vec{y}.u$ and $\vec{y}.v$ be terms where the latter one is expanded
as well as a substitution
$\theta = \SET{\vec{y} \mapsto \vec{s}}$ where $\vec[k]{x}.s_i$ is an
expanded term for all $1 \leq i \leq n$. If
$\vec{y}.u \subtEeq \vec{y}.v$ then
$\vec[k]{x}.u\theta \subtEeq \vec[k]{x}.v\theta$.
\end{lemma}

\begin{proof}
We proceed by induction on $\vec{y}.u$. Let
$\vec{y}.u = \vec[n']{y}.h_1(\vec[m_1]{u})$ and
$\vec{y}.v = \vec{y},\vec[k]{z}.h_2(\vec[m_2]{v},\vec[k]{z}\cf)$. By
definition, there exist $n'' \geq n'$ and $k$ terms written as
$\vec[n'']{y}.v_{m_2+1},\dots,\vec[n'']{y}.v_{m_2+k}$ such that
$\vec{y}.u \subteq \vec[n'']{y}.h_2(\vec[m_2+k]{v})$.
Without loss of generality,
$\SEt{y_{n+1},\dots,y_{n''}}$ is fresh for $\theta$.
If the two terms are equal then this
is preserved by substitution application, so
$\vec[k]{x}.u\theta \subtEeq \vec[k]{x}.v\theta$ as desired. Otherwise,
$\vec{y}.u_i \subteq \vec[n'']{y}.h_2(\vec[m_2+k]{v})$ for some
$1 \leq i \leq m_1$ and the induction hypothesis yields
$\vec[k]{x}.u_i\theta \subtEeq \vec[k]{x}.v\theta$. If
$h_1 \notin \dom{\theta}$ then
$\vec[k]{x}.u\theta \subtEeq \vec[k]{x}.v\theta$. Otherwise, let
$h_1 = y_i$ for some $1 \leq i \leq n$ and
$\vec{x}.s_i =
\vec{x},\vec[m_1]{z}.h_i(\vec[l]{{s_i}},\vec[m_1]{z}\cf)$. Then,
$\vec[k]{x}.u\theta =
\vec[k]{x}.h_i(\vec[l]{{s_i}},\vec[m_1]{u}\theta)$, so we
again have $\vec[k]{x}.u\theta \subtEeq \vec[k]{x}.v\theta$. \qed
\end{proof}

The following two lemmata incrementally build up to the proof
of \lemref{dhpsubst}, the remaining auxiliary result for \thmref{dhpu}.

\begin{lemma}
\label{lem:dhpVarArgListSubst}
If $\vec[k]{\vec[m]{y},\vec[l]{z}.u}$ and $\vec[m]{\vec{x}.s}$ are \textup{DHP}
var-arg lists then
\[\vecp[k]{\vec{x},\vec[l]{z}.u}{\SEt{\vec[m]{y} \mapsto \vec[m]{s}}}\]
is a \textup{DHP} var-arg list.
\end{lemma}

\begin{proof}
Let $\mu = \SEt{\vec[m]{y} \mapsto \vec[m]{s}}$. Without loss
of generality, $\SEt{\vec[l]{z}}$ is fresh for $\mu$.
Consider a term
$\vec{x},\vec[l]{z}.u_i\mu$ where $1 \leq i \leq k$. By definition,
$\varnothing \neq \fv{u_i} \subseteq \SEt{\vec[m]{y}} \cup \SEt{\vec[l]{z}}$ and
$\varnothing \neq \fv{s_l} \subseteq \SEt{\vec{x}}$ for all
$1 \leq l \leq m$. Hence, we have
$\varnothing \neq \fv{u_i\mu} \subseteq \SEt{\vec{x}} \cup \SEt{\vec[l]{z}}$.
Furthermore, if $\vec[m]{y},\vec[l]{z}.u_i$ is an expanded term then
$\vec{x},\vec[l]{z}.u_i\mu$ is an expanded term as substitution application is
capture-avoiding. Hence, we are left to show that for any term
$\vec{x},\vec[l]{z}.u_j\mu$ with $i \neq j$ we have
$\vec{x},\vec[l]{z}.u_i\mu \nsubtEeq \vec{x},\vec[l]{z}.u_j\mu$. We perform induction on
$\vec{x},\vec[l]{z}.u_i\mu$. By definition,
$\vec{x},\vec[l]{z}.u_i = \vec{x},\vec[l]{z},\vec[l_1]{w}.h_1(\vec[c_1]{v},\vec[l_1]{w}\cf)$ and
$\vec{x},\vec[l]{z}.u_j = \vec{x},\vec[l]{z},\vec[l_2]{w}.h_2(\vec[c_2]{t},\vec[l_2]{w}\cf)$.
Let $l_3 = \m{max}(l_1,l_2)$.
Without loss of generality,
$\SEt{\vec[l_3]{w}}$ is fresh for $\mu$ and
$\SEt{\vec[l_3]{w}} \cap \SEt{\vec[l]{z}} = \varnothing$.
We proceed by case analysis on $h_1$.
\begin{itemize}
\item
If $h_1 \in \xF \cup \SEt{\vec[l]{z}}$ then $\vec{x},\vec[l]{z}.u_i\mu =
\vec{x},\vec[l]{z},\vec[l_1]{w}.h_1(\vec[c_1]{v}\mu,\vec[l_1]{w}\cf)$. By
the induction hypothesis,
$\vec{x},\vec[l]{z},\vec[l_1]{w}.v \nsubtEeq \vec{x},\vec[l]{z}.u_j\mu$ for all
$v \in \SET{\vec[c_1]{v}\mu,\vec[l_1]{w}\cf}$.
\begin{itemize}
\smallskip
\item
If also $h_2 \in \xF \cup \SEt{\vec[l]{z}}$ then $\vec{x},\vec[l]{z}.u_j\mu =
\vec{x},\vec[l]{z},\vec[l_2]{w}.h_2(\vec[c_2]{t}\mu,\vec[l_2]{w}\cf)$.
Clearly, $\vec{x},\vec[l]{z}.u_i\mu \nsubtEeq \vec{x},\vec[l]{z}.u_j\mu$ if
$c_2 > c_1$ or $h_1 \neq h_2$. Hence, we may assume that
$c_2 \leq c_1$ as well as $h_1 = h_2$ and are left to prove that
$\vec{x},\vec[l]{z}.v_i\mu \neq \vec{x},\vec[l]{z}.t_i\mu$ for some
$1 \leq i \leq c_2$. From $\vec{x},\vec[l]{z}.u_i \nsubtEeq \vec{x},\vec[l]{z}.u_j$ we
obtain $\vec{x},\vec[l]{z}.v_i \neq \vec{x},\vec[l]{z}.t_i$ for some
$1 \leq i \leq c_2$, so the desired result follows from
injectivity of $\mu$.
\item
Let $h_2 = y_{j'}$ where $1 \leq j' \leq m$ and
$\vec{x}.s_{j'} =
\vec{x},\vec[c_2+l_2]{w}.h_{j'}(\vec[d]{{s_{j'}}},\vec[c_2+l_2]{w}\cf)$.
Then,
$\vec{x},\vec[l]{z}.u_j\mu = \vec{x},\vec[l]{z},\vec[l_2]{w}.h_{j'}(\vec[d]{{s_{j'}}},
\vec[c_2]{t}\mu,\vec[l_2]{w}\cf)$.
Clearly, $\vec{x},\vec[l]{z}.u_i\mu \nsubtEeq \vec{x},\vec[l]{z}.u_j\mu$ if
$d+c_2 > c_1$ or $h_1 \neq h_{j'}$. Hence, we may assume that
$d+c_2 \leq c_1$ as well as $h_1 = h_{j'}$ and are left to prove
that the two terms differ in at least one of the first $d+c_2$
arguments of the head symbol. For a proof by contradiction,
assume $\vec{x}.v_i\mu = \vec{x}.{s_{j'}}_i$ for all
$1 \leq i \leq d$. Since $\SEt{\vec[l]{z}}$ is fresh for $\mu$,
$h_1 = h_{j'} \in \xF$, so we know that there
is at least some $1 \leq o \leq d$ such that
$\varnothing \neq \fv{{s_{j'}}_o}$. This implies that
$\vec{x}.{s_{j'}}_o = \vec{x}.v_o\mu \subtEeq \vec{x}.s_{o'}$
for some $1 \leq o' \leq m$. Together with
$\vec{x}.s_{j'} \subt \vec{x}.{s_{j'}}_o$ we obtain
$\vec{x}.s_{j'} \subtE \vec{x}.s_{o'}$ which results in a
contradiction with our assumptions regardless of the
relationship between $j'$ and $o'$.
\end{itemize}
\item
If $h_1 = y_{i'}$ for some
$1 \leq i' \leq m$ then $\vec{x}.s_{i'} =
\vec{x},\vec[c_1+l_1]{w}.h_{i'}(\vec[d]{{s_{i'}}},\vec[c_1+l_1]{w}\cf)$
and
$\vec{x},\vec[l]{z}.u_i\mu = \vec{x},\vec[l]{z},\vec[l_1]{w}.h_{i'}(\vec[d]{{s_{i'}}},
\vec[c_1]{v}\mu,\vec[l_1]{w}\cf)$.
By the induction hypothesis,
$\vec{x},\vec[l]{z},\vec[l_1]{w}.v \nsubtEeq \vec{x},\vec[l]{z}.u_j\mu$ for all
$v \in \SET{\vec[d]{{s_{i'}}},\vec[c_1]{v}\mu,\vec[l_1]{w}\cf}$.
\begin{itemize}
\smallskip
\item
If $h_2 \in \xF \cup \SEt{\vec[l]{z}}$ then
$\vec{x},\vec[l]{z}.u_j\mu =
\vec{x},\vec[l]{z},\vec[l_2]{w}.h_2(\vec[c_2]{t}\mu,\vec[l_2]{w}\cf)$. Clearly,
$\vec{x},\vec[l]{z}.u_i\mu \nsubtEeq \vec{x},\vec[l]{z}.u_j\mu$ if $c_2 > d+c_1$ or
$h_{i'} \neq h_2$. Hence, we may assume that $c_2 \leq d+c_1$ as
well as $h_{i'} = h_2$ and are left to prove that the two terms
differ in at least one of the first $c_2$ arguments of the
head symbol. Since $\SEt{\vec[l]{z}}$ is fresh for $\mu$,
$h_{i'} = h_2 \in \xF$, so we know
that there is some minimal $1 \leq o \leq d$ such
that $\varnothing \neq \fv{{s_{i'}}_o}$.
First assume $o \leq c_2$. If
$\vec{x}.{s_{i'}}_o \neq \vec{x}.t_o\mu$, we are done immediately.
Otherwise, $\vec{x}.{s_{i'}}_o = \vec{x}.t_o\mu \subtEeq \vec{x}.s_{o'}$
for some $1 \leq o' \leq m$. Together with
$\vec{x}.s_{i'} \subt \vec{x}.{s_{i'}}_o$ we obtain
$\vec{x}.s_{i'} \subtE \vec{x}.s_{o'}$ which results in a
contradiction with our assumptions regardless of the
relationship between $i'$ and $o'$.
Now assume $o > c_2$. For a proof by contradiction,
assume $\vec[c_2]{\vec{x}.s_{i'}} = \vec[c_2]{\vec{x},\vec[l]{z}.t\mu}$.
By minimality of $o$ together with the definition of $\mu$, this means that
$\fv{\vec[c_2]{t}} = \varnothing$ and therefore also
$\fv{u_j} = \varnothing$ which contradicts our assumption that
$\vec[k]{\vec[m]{y},\vec[l]{z}.u}$ is a DHP var-arg list.
\item Let $h_2 = y_{j'}$ where $1 \leq j' \leq m$ and
$\vec{x}.s_{j'} =
\vec{x},\vec[c_2+l_2]{w}.h_{j'}(\vec[e]{{s_{j'}}},\vec[c_2+l_2]{w}\cf)$.
Then,
$\vec{x},\vec[l]{z}.u_j\mu = \vec{x},\vec[l]{z},\vec[l_2]{w}.h_{j'}(\vec[e]{{s_{j'}}},
\vec[c_2]{t}\mu,\vec[l_2]{w}\cf)$.
Clearly, $\vec{x},\vec[l]{z}.u_i\mu \nsubtEeq \vec{x},\vec[l]{z}.u_j\mu$ if
$e+c_2 > d+c_1$ or $h_{i'} \neq h_{j'}$. Hence, we may assume
that $e+c_2 \leq d+c_1$ as well as $h_{i'} = h_{j'}$ and are
left to prove that the two terms differ in at least one of the
first $e+c_2$ arguments of the head symbol. If $i' = j'$ then
$d=e$ and thus $c_2 \leq c_1$. From
$\vec{x},\vec[l]{z}.u_i \nsubtEeq \vec{x},\vec[l]{z}.u_j$ we obtain
$\vec{x},\vec[l]{z}.v_i \neq \vec{x},\vec[l]{z}.t_i$ for some $1 \leq i \leq c_2$, so
the desired result follows from injectivity of $\mu$. Now
consider the case $i' \neq j'$. For a proof by contradiction,
assume $\vec{x}.{s_{i'}}_i = \vec{x}.{s_{j'}}_i$ for all
$1 \leq i \leq \m{min}(d,e)$. In particular, this implies $d \neq e$.
If $d > e$ then
$\vec{x}.s_{i'} \subtEeq \vec{x}.s_{j'}$ and if $e > d$ then
$\vec{x}.s_{j'} \subtEeq \vec{x}.s_{i'}$. Hence, we arrive at a
contradiction with our initial assumptions. \qed
\end{itemize}
\end{itemize}
\end{proof}

\begin{lemma}
\label{lem:dhpsubstrec}
If $\vec[m]{y}.u\in \dhp$ then
$\vec{x}.u\SEt{\vec[m]{y} \mapsto \vec[m]{s}} \in \dhp$
whenever $\vec[m]{\vec{x}.s}$ is a \textup{DHP} var-arg list.
\end{lemma}

\begin{proof}
Let $\mu = \SEt{\vec[m]{y} \mapsto \vec[m]{s}}$. We proceed by
induction on $\vec[m]{y}.u$, so consider
$\vec[m]{y}.u \subteq \vec[m]{y},\vec[l]{z}.h(\vec[k]{u})$. Without loss
of generality, $\SEt{\vec[l]{z}}$ is fresh for $\mu$.
By the induction
hypothesis, $\vec{x},\vec[l]{z}.u_i\mu \in \dhp$ for all $1 \leq i \leq k$. If
$h \in \xF \cup \SEt{\vec[l]{z}}$ then we immediately conclude that
$\vec{x},\vec[l]{z}.h(\vec[k]{u})\mu \in \dhp$. Now consider the case where
$h = y_j$ for some $1 \leq j \leq m$. Note that
$\vec{x}.s_j = \vec{x},\vec[k]{w}.h(\vec[l]{v},\vec[k]{w}\cf)$ is an
expanded term which does not contain free variables. Hence,
together with the induction hypothesis we obtain
$\vec{x},\vec[l]{z}.y_j(\vec[k]{u})\mu =
\vec{x},\vec[l]{z}.h(\vec[l]{v},\vec[k]{u}\mu)
\in \dhp$.
Finally, let
$h \in \xV \setminus (\SEt{\vec[m]{y}} \cup \SEt{\vec[l]{z}})$.
In particular, $h \not\in \dom{\mu}$ and
$\vec[k]{\vec[m]{y},\vec[l]{z}.u}$ is a
DHP var-arg list. Hence, the result follows from
\lemref{dhpVarArgListSubst}. \qed
\end{proof}

\begin{lemma}
\label{lem:dhpsubst}
If $s \in \dhp$ and $\theta$ is a substitution such that
$\im{\theta} \subseteq \dhp$ then $s\theta \in \dhp$.
\end{lemma}

\begin{proof}
We proceed by induction on $s$, so let $s = \vec{x}.h(\vec[m]{s})$.
By the induction hypothesis, $(\vec{x}.s_i)\theta \in \dhp$ for
$1 \leq i \leq m$. If $h \in \xF \cup \SEt{\vec{x}}$ we immediately
conclude $s\theta \in \dhp$. Otherwise,
$h = x \in \xV \setminus \SEt{\vec{x}}$. Since $s \in \dhp$,
$\fv{\vec{x}.s_i} = \varnothing$ and therefore
$(\vec{x}.s_i)\theta = \vec{x}.s_i$ for all $1 \leq i \leq m$.
Hence, if $x \not\in \dom{\theta}$ then $s\theta \in \dhp$. Now
assume $\theta(x) = \vec[m]{y}.u \in \dhp$.
Without loss of generality we can choose
$\vec{x}$ and $\vec[m]{y}$ such that
$\SEt{\vec{x}}$ is fresh for $\theta$
and $\SEt{\vec{x}} \cap \SEt{\vec[m]{y}} = \varnothing$, so
$(\vec{x}.x(\vec[m]{s}))\theta = \vec{x}.u\SET{\vec[m]{y} \mapsto
\vec[m]{s}\theta} = \vec{x}.u\SEt{\vec[m]{y} \mapsto
\vec[m]{s}}$. Hence, \lemref{dhpsubstrec} yields
$(\vec{x}.x(\vec[m]{s}))\theta \in \dhp$ as desired. \qed
\end{proof}

Finally, we can prove \thmref{dhpu}.

\begin{proof}[of \thmref{dhpu}]
We perform case analysis on the inference rule applied in the step.
The result is trivial for \urem and \udec. For the remaining rules,
we can use \lemref{dhpsubst} as the substitutions used in them only
map variables to DHPs: This is true by assumption for \uvar.
Furthermore, partial bindings are DHPs by definition, so \uimt and
\uprj are also covered. Moreover, the term $u$ used in \uffe is
also a DHP. Finally, we have to show that the terms $u$ and $v$ used
in \uffne are DHPs: From $P_1$ and $P_2$, we can see that conditions
\enuref{dhp-var} and \enuref{dhp-exp}
of \defref{dhp} clearly hold for all $\vec{y}.u_i$ and $\vec[m]{z}.v_i$
where $1 \leq i \leq l$. For a proof by contradiction regarding
condition \enuref{dhp-subt}, assume
$\vec{y}.u_i \subtEeq \vec{y}.u_j$ for some $i \neq j$. If both
terms come from $P_2$ then
$\vec[k]{x}.u_i\SEt{\vec{y} \mapsto \vec{s}} = \vec{x}.t_{i'}$ for
some $1 \leq i' \leq m$ and
$\vec[k]{x}.u_j\SEt{\vec{y} \mapsto \vec{s}} = \vec{x}.t_{j'}$ for
some $1 \leq j' \leq m$ where $i' \neq j'$. Together with
\lemref{subtEclosedUnderSubst} we obtain
$\vec[k]{x}.t_{i'} \subtEeq \vec[k]{x}.t_{j'}$ for $i' \neq j'$
which contradicts our assumption $E \subseteq \dhp \times \dhp$. If
$\vec{y}.u_i = \vec{y}.y_{i'}\cf$ for some $i'$ then
$\vec{y}.u_j = \vec{y}.y_{i'}\cf$ is the only possibility.
Hence, we only have to consider one remaining case where
$\vec{y}.u_j = \vec{y}.y_{j'}\cf$ for some $j'$ and
$y_{j'} \in \fv{u_i}$. By definition, there is a
$\vec[k]{x}.t_{l_1}$ such that
$\vec[k]{x}.s_{j'} \subtEeq \vec[k]{x}.t_{l_1}$ ($P_1$ for $u_j$) and a
$\vec[k]{x}.t_{l_2}$ such that
$\vec[k]{x}.t_{l_2} \subtEeq \vec[k]{x}.s_{j'}$ ($P_2$ for $u_i$). From
\lemref{subtEtrans} we obtain
$\vec[k]{x}.t_{l_2} \subtEeq \vec[k]{x}.t_{l_1}$. If
$l_1 \neq l_2$ we arrive at a contradiction with
$E \subseteq \dhp \times \dhp$. Otherwise,
$\vec[k]{x}.t_{l_1} = \vec[k]{x}.t_{l_2}$ and therefore
$\vec[k]{x}.t_{l_2} = \vec[k]{x}.s_{j'} = \vec[k]{x}.t_{l_1}$ which means
that $\vec{y}.u_i = \vec{y}.u_j$,
$\vec[m]{z}.v_i = \vec[m]{z}.v_j$ and therefore $i = j$ as the pairs of
arguments are taken from a set. The proof of condition \enuref{dhp-subt}
for $\vec[l]{\vec[m]{z}.v}$ is analogous. \qed
\end{proof}

\section{Proofs of Lemmata \ref{lem:nondetComplete} and
  \ref{lem:custepTerm}}
\label{app:lemmata}

We start
by establishing a sufficient condition for
steps in $\custep$ will often be used in the
proof of \lemref{nondetComplete}.

\begin{lemma}
\label{lem:ndcAux}
Consider a unification pair $(E,\theta)$, a set of variables $W$
and an idempotent substitution $\delta \in \unifiers{E}$
such that $\theta\delta = \delta\restr{W}$.
Let $\delta' = \delta \uplus \mu$
where $\fv{\mu} \subseteq \fv{\delta}$ and
$\dom{\mu} \subseteq \fv{\gamma}$ is a set of variables not occurring
elsewhere. If
$(E,\theta) \ustep (E'\gamma,\theta\gamma)$, $E' \subseteq E$ and
$\gamma\delta' = \delta\restr{\dom{\gamma}}$ then
$(E,\theta,\delta) \custep (E'\gamma,\theta\gamma,\delta')$.
\end{lemma}

\begin{proof}
We start by proving $\gamma\delta' = \delta \restr{V}$ for all sets
of variables $V$ such that $V \cap \dom{\mu} = \varnothing$
(\textasteriskcentered). Let $x \in V$, so $x \notin \dom{\mu}$. If
$x \in \dom{\gamma}$ then $\gamma(x)\delta' = \delta(x)$ by
assumption. Otherwise, $x \notin \dom{\gamma}$. Since
$x \notin \dom{\mu}$, $\gamma(x)\delta' = \delta(x)$ as desired.
From $\delta \in \unifiers{E}$ we infer $u\delta = v\delta$ for all
$u \approx v \in E$. Together with (\textasteriskcentered) we obtain
$u\gamma\delta' = v\gamma\delta'$ as
$\fv{E} \cap \dom{\mu} = \varnothing$ and therefore
$u'\delta' = v'\delta'$ for all $u' \approx v' \in E'\gamma$ which
means that $\delta' \in \unifiers{E'\gamma}$. Furthermore, we have
$\theta\gamma\delta' = \theta\delta = \delta = \delta'\restr{W}$
where the first equality follows from (\textasteriskcentered)
because $(W \cup \fv{\theta}) \cap \dom{\mu} = \varnothing$.
Finally, idempotence of $\delta'$ follows of idempotence of
$\delta$ together with $\fv{\mu} \subseteq \fv{\delta$} and
$\dom{\mu} \cap \fv{\delta} = \varnothing$.
\qed
\end{proof}

Next comes a result which is needed for the \uffne case in
the proof of \lemref{nondetComplete}.

\begin{lemma}
\label{lem:ffneComplete}
Let $E$ be a set of unordered pairs of \textup{DHPs},
$\vec[k]{x}.F(\vec{s}) \approx \vec[k]{x}.G(\vec[m]{t}) \in E$ and
$\delta \in \unifiers{E}$ such that
$\SET{F \mapsto \vec{y}.p, G \mapsto \vec[m]{z}.q} \subseteq \delta$.
If $\vec[l]{\vec{y}.u}$ and $\vec[l]{\vec[m]{z}.v}$ are defined as in
\textup{\uffne} then there exists a term $\vec[l]{w}.u$ such that
\begin{enumerate}[(i)]
\item \label{enu:ffneCompleteF}
$\vec{y}.u\SEt{\vec[l]{w} \mapsto \vec[l]{u}} = \vec{y}.p$,
\item \label{enu:ffneCompleteG}
$\vec[m]{z}.u\SEt{\vec[l]{w} \mapsto \vec[l]{v}} = \vec[m]{z}.q$.
\end{enumerate}
\end{lemma}

\begin{proof}
Since $\delta \in \unifiers{E}$,
$\vec[k]{x}.p\SEt{\vec{y} \mapsto \vec{s}} =
\vec[k]{x}.q\SEt{\vec[m]{z} \mapsto \vec[m]{t}}$
(\textasteriskcentered). We prove that (\textasteriskcentered)
implies \enuref{ffneCompleteF} and \enuref{ffneCompleteG}
simultaneously by induction on $\vec{y}.p$ and $\vec[m]{z}.q$,
respectively. Let $\vec{y}.p = \vec{y}.h_1(\vec[o_1]{p})$ and
$\vec[m]{z}.q = \vec[m]{z}.h_2(\vec[o_2]{q})$. If
$h_1 \in \xF \cup (\xV \setminus \SEt{\vec{y}})$ and
$h_2 \in \xF \cup (\xV \setminus \SEt{\vec[m]{z}})$ then $h_1 = h_2$
and therefore also $o_1 = o_2$ follows from (\textasteriskcentered).
Furthermore, (\textasteriskcentered) yields
$\vec[k]{x}.p_i\SEt{\vec{y} \mapsto \vec{s}} =
\vec[k]{x}.q_i\SEt{\vec[m]{z} \mapsto \vec[m]{t}}$
for all $1 \leq i \leq o_1$. Therefore, we
obtain the existence of $\vec[o_1]{\vec[l]{w}.r}$ such that
$\vec{y}.r_i\SEt{\vec[l]{w} \mapsto \vec[l]{u}} = \vec{y}.p_i$ and
$\vec[m]{z}.r_i\SEt{\vec[l]{w} \mapsto \vec[l]{v}} = \vec[m]{z}.q_i$
for all $1 \leq i \leq o_1$ from the induction hypotheses. Hence,
we define $\vec[l]{w}.u = \vec[l]{w}.h_1(\vec[o_1]{r})$ which
satisfies both \enuref{ffneCompleteF} and \enuref{ffneCompleteG}.
Now consider the case where $h_1 \in y_i$ for some $1 \leq i \leq n$,
so by (\textasteriskcentered) we have $o_2 \geq o_1$. Now let
$r = \vec[m]{z},\vec[o_1]{y}.h_2(\vec[o_2-o_1]{q},\vec[o_1]{y}\cf)$
where we choose $\vec[o_1]{y}$ such that
$(\fv{\vec[o_2-o_1]{q}} \cup \SEt{h_2}) \cap \SEt{\vec[o_1]{y}} =
\varnothing$. From (\textasteriskcentered) and the fact that
$\vec[k]{x}.s_i$ is an expanded term we obtain
$\vec[k]{x}.s_i = \vec[k]{x}.r\SEt{\vec[m]{z} \mapsto \vec[m]{t}}$
as well as
$\vec[k]{x}.p_i\SEt{\vec{y} \mapsto \vec{s}} =
\vec[k]{x}.q_{o_2-o_1+i}\SEt{\vec[m]{z} \mapsto \vec[m]{t}}$ for all
$1 \leq i \leq o_1$. Hence, there is some $j$ such that
$\vec{y}.y_i = \vec{y}.u_j$ and $\vec[m]{z}.r = \vec[m]{z}.v_j$.
Moreover, the induction hypotheses yield the existence of
$\vec[o_1]{\vec[l]{w}.r}$ such that
$\vec{y}.r_i\SEt{\vec[l]{w} \mapsto \vec[l]{u}} = \vec{y}.p_i$ and
$\vec[m]{z}.r_i\SEt{\vec[l]{w} \mapsto \vec[l]{v}} =
\vec[m]{z}.q_{o_2-o_1+i}$ for all $1 \leq i \leq o_1$. We define
$\vec[l]{w}.u = \vec[l]{w}.w_j(\vec[o_1]{r})$ which satisfies both
\enuref{ffneCompleteF} and \enuref{ffneCompleteG}.
The case where $h_2 \in \SEt{\vec[m]{z}}$ is symmetrical. \qed
\end{proof}

Finally, we can prove the two lemmata.

\begin{proof}[of \lemref{nondetComplete}]
Assume $E \neq \varnothing$, so $E = \SET{s \approx t} \uplus E_0$.
If $s = t$ then we can apply \urem to obtain
$(E,\theta,\delta) \custep (E_0,\theta,\delta)$.
Otherwise, let
$s = \vec[k]{x}.h_1(\vec{s})$ and $t = \vec[k]{x}.h_2(\vec[m]{t})$.
Since $\delta \in \unifiers{E}$, $s\delta = t\delta$. Now assume
$h_1 = h_2$, so $n = m$. If
$h_1 = h_2 \in \xF \cup \SEt{\vec[k]{x}}$ then we can apply \udec
to obtain
$(E,\theta,\delta) \custep
 (E_0 \cup \bigcup_{1 \leq i \leq n}\SEt{\vec[k]{x}.s_i \approx \vec[k]{x}.t_i},\theta,\delta)$.
Otherwise,
$h_1 = h_2 \in \xV \setminus \SEt{\vec[k]{x}}$ and we can apply
\uffe using a fresh variable $H \notin W$, so
$(E,\theta) \ustep (E_0\gamma,\theta\gamma)$ where
$\gamma = \SET{h_1 \mapsto \vec{y}.H(y_{i_1}\cf,\dots,y_{i_l}\cf)}$.
Since $\delta \in \unifiers{E}$, $h_1 \mapsto \vec{y}.u \in \delta$
for some term $\vec{y}.u$ and we define
$\delta' = \delta \uplus \mu$ where
$\mu = \SET{H \mapsto y_{i_1},\dots,y_{i_l}.u}$. Note that
$\fv{u} \cap \SEt{\vec{y}} \subseteq \SEt{y_{i_1},\dots,y_{i_l}}$ since
$\vec[k]{x}.u\SEt{\vec{y} \mapsto \vec{s}} =
\vec[k]{x}.u\SEt{\vec{y} \mapsto \vec{t}}$, so we cannot use $y_j$
such that $\vec[k]{x}.s_j \neq \vec[k]{x}.t_j$ in $u$ as all $y_i$'s
are mapped to expanded terms. Since
$\fv{\mu} \subseteq \fv{\delta}$,
$\dom{\mu} \subseteq \fv{\gamma}$ and
$\gamma(h_1)\delta' = \gamma(h_1)\mu = \vec{y}.u = \delta(h_1)$,
\lemref{ndcAux} yields
$(E,\theta,\delta) \custep (E_0\gamma,\theta\gamma,\delta')$.
Next, assume
$h_1 \neq h_2$. Since $s\delta = t\delta$, at least one of the two
heads has to be a free variable. Without loss of generality, let
$h_1 \in \xV \setminus \SEt{\vec[k]{x}}$. If $h_1 \notin \fv{t}$
then we can apply \uvar to obtain
$(E,\theta) \ustep (E_0\gamma,\theta\gamma)$ where
$\gamma = \SET{h_1 \mapsto t}$ and set $\delta' = \delta$, so
$\delta'$ is also idempotent. Since $\delta \in \unifiers{E}$,
$\gamma = \delta\restr{\dom{\gamma}}$. Together with idempotence of
$\delta'$ we obtain $\gamma\delta' = \delta\restr{\dom{\gamma}}$.
Thus, an application of \lemref{ndcAux} for $\mu = \varnothing$
yields $(E,\theta,\delta) \custep (E_0\gamma,\theta\gamma,\delta')$.
Now assume
$h_2 \in \xF \cup \SEt{\vec[k]{x}}$. Since $s\delta = t\delta$ we
have $h_1 \mapsto \vec{y}.h_3(\vec[o]{v}) \in \delta$ where
$h_3 \in \xF \cup \SEt{\vec{y}}$ as either $h_3 = h_2 \in \xF$ or
$h_3 = y_i$ and $h_2 = \hd{\vec[k]{x}.s_i}$. Hence, for the fresh
variables $\seq[o]{H} \notin W$ we have
$u = \vec{y}.h_3(\vec[k_1]{z^1}.H_1(\vec{y}\cf,\vec[k_1]{z^1}\cf),
\dots,\vec[k_o]{z^o}.H_o(\vec{y}\cf,\vec[k_o]{z^o}\cf))
\in \ib{\tau} \cup \pb{\tau}$, so we can apply \uimt or \uprj to
obtain $(E,\theta) \ustep (E\gamma,\theta\gamma)$ where
$\gamma = \SET{h_1 \mapsto u}$. We define
$\delta' = \delta \uplus \mu$ where
$\mu = \SET{H_i \mapsto \vec{y}.v_i \mid 1 \leq i \leq o}$.
By the type of $h_3$, we
have $\vec{y}.v_i = \vec{y},\vec[k_i]{z^i}.v_i'$ for all
$1 \leq i \leq o$. Hence,
$(\vec{y},\vec[k_i]{z^i}.H_i(\vec{y}\cf,\vec[k_i]{z^i}\cf))\mu = \vec{y}.v_i$ for
all $1 \leq i \leq o$ and therefore
$\gamma(h_1)\delta' = \gamma(h_1)\mu = (\vec{y}.u)\mu =
\vec{y}.h_3(\vec[o]{v}) = \delta(h_1)$. Since also
$\fv{\mu} \subseteq \fv{\delta}$ and $\dom{\mu} \subseteq \fv{\gamma}$,
we can use \lemref{ndcAux} to obtain
$(E,\theta,\delta) \custep (E_0\gamma,\theta\gamma,\delta')$.
Finally, we are
left with the case where $h_2 \in \xV \setminus \SEt{\vec[k]{x}}$,
so we can apply \uffne using a fresh variable $H \notin W$, and
$(E,\theta) \ustep (E_0\gamma,\theta\gamma)$ where
$\gamma = \SET{h_1 \mapsto \vec{y}.H(\vec[l]{u}),
h_2 \mapsto \vec[m]{z}.H(\vec[l]{v})}$. Since $\delta \in \unifiers{E}$,
$\SET{h_1 \mapsto \vec{y}.p, h_2 \mapsto \vec[m]{z}.q} \subseteq
\delta$ for some terms $ \vec{y}.p$ and
$\vec[m]{z}.q$. \lemref{ffneComplete} yields a term $\vec[l]{w}.u$
such that $\fv{\vec[l]{w}.u} \subseteq \fv{\delta}$,
$\vec{y}.u\SEt{\vec[l]{w} \mapsto \vec[l]{u}} = \vec{y}.p$ and
$\vec[m]{z}.u\SEt{\vec[l]{w} \mapsto \vec[l]{v}} = \vec[m]{z}.q$, so
we define $\delta' = \delta \uplus \mu$ where
$\mu = \SET{H \mapsto \vec[l]{w}.u}$. By definition,
$\gamma(h_1)\delta' = \gamma(h_1)\mu = \delta(h_1)$ and
$\gamma(h_2)\delta' = \gamma(h_2)\mu = \delta(h_2)$.
Since also
$\fv{\mu} \subseteq \fv{\delta}$ and $\dom{\mu} \subseteq \fv{\gamma}$,
we can use \lemref{ndcAux} to obtain
$(E,\theta,\delta) \custep (E_0\gamma,\theta\gamma,\delta')$.
\qed
\end{proof}

\begin{proof}[of \lemref{custepTerm}]
Consider the measure $\m{m}(E,\theta,\delta) = (m,n)$ where
\begin{align*}
m &= \sum\SET{|\delta(x)| \mid x \in \fv{E}} &
n &= \sum\SET{|s| + |t| \mid s \approx t \in E}
\end{align*}
If $(E,\theta,\delta) \custep (E',\theta',\delta')$ then
$\m{m}(E,\theta,\delta) = (m,n) >_\lex (m',n') =
\m{m}(E',\theta',\delta')$: If the underlying unification inference
rule was \urem or \udec, $m \geq m'$ and $n > n'$. For \uvar we have
$m > m'$ because the variable does not occur free in $E'$ anymore.
Furthermore, in the case of \uimt and \uprj we have $m > m'$
because the sum
of the size of terms to which the newly introduced variables are
mapped is smaller than the size of the term to which the processed
variable was mapped
($|\vec{y}.h_3(\vec[o]{v})| > \sum\SET{|\vec{y}.v_i| \mid 1 \leq i \leq o}$
using the notation from the proof of \lemref{nondetComplete}).
For
\uffe, $m = m'$ but $n > n'$. For \uffne we have $m > m'$
because we remove two free variables and add only one new variable
which is mapped to a term which is not bigger than the old mappings
($\fv{E'} = \fv{E} \cup \SET{H} \setminus \SEt{h_1,h_2}$ and
$|\vec{y}.p| \geq |\vec[l]{w}.u|$ using the notation from the proof
of \lemref{nondetComplete}). \qed
\end{proof}

\end{document}